\UseRawInputEncoding

\documentclass[12pt]{article}
\usepackage{longtable,supertabular}
\usepackage{amsmath}
\usepackage{amssymb}
\usepackage{latexsym}
\usepackage{cite,a4wide,color}
\usepackage{multirow}
\usepackage{ntheorem}
\usepackage[square,numbers]{natbib}

\newtheorem{theorem}{Theorem}[section]
\newtheorem{corollary}{Corollary}[section]
\newtheorem{proposition}{Proposition}[section]
\newtheorem{lemma}{Lemma}[section]

\newtheorem*{proof}{Proof.}

\newtheorem{example}{Example}[section]

\theoremstyle{remark}

\title{\bf Linear Codes from Projective Linear Anticodes Revisited}
\author{Hao Chen and Conghui Xie\thanks{Hao Chen and Conghui Xie are with the College of Information Science and Technology/Cyber Security, Jinan University, Guangzhou, Guangdong Province, 510632, China, haochen@jnu.edu.cn, conghui@stu2021.jnu.edu.cn. This research was supported by NSFC Grant 62032009.
}}

\begin{document}
\maketitle
\begin{abstract}
An anticode ${\bf C} \subset {\bf F}_q^n$ with the diameter $\delta$ is a code in ${\bf F}_q^n$ such that the distance between any two distinct codewords in ${\bf C}$ is at most $\delta$.  The famous Erd\"{o}s-Kleitman bound for a binary anticode ${\bf C}$ of the length $n$ and the diameter $\delta$ asserts that $$|{\bf C}| \leq \Sigma_{i=0}^{\frac{\delta}{2}} \displaystyle{n \choose i}.$$ In this paper, we give an antiGriesmer bound for $q$-ary projective linear anticodes, which is stronger than the above Erd\"{o}s-Kleitman bound for binary anticodes. The antiGriesmer bound is a lower bound on diameters of projective linear anticodes. From some known projective linear anticodes, we construct some linear codes with optimal or near optimal minimum distances. A complementary theorem constructing infinitely many new projective linear $(t+1)$-weight code from a known $t$-weight linear code is presented. Then many new optimal or almost optimal few-weight linear codes are given and their weight distributions are determined. As a by-product, we also construct several infinite families of three-weight binary linear codes, which lead to $l$-strongly regular graphs for each odd integer $l \geq 3$.\\

 Optimal or almost optimal few-weight linear code. Complementary Reed-Solomon code. $l$-strongly walk-regular graph.
\end{abstract}

\newpage

\section{Introduction}\label{sec-1}

\subsection{Preliminaries}

The Hamming weight $wt({\bf a})$ of a vector ${\bf a}=(a_0, \ldots, a_{n-1}) \in {\bf F}_q^n$ is the cardinality of its support, $$supp({\bf a})=\{i: a_i \neq 0\}.$$ The Hamming distance $d({\bf a}, {\bf b})$ between two vectors ${\bf a}$ and ${\bf b}$ is $d({\bf a}, {\bf b})=wt({\bf a}-{\bf b})$. Then ${\bf F}_q^n$ is a finite Hamming metric space. The minimum (Hamming) distance of a code ${\bf C} \subset {\bf F}_q^n$ is, $$d({\bf C})=\min_{{\bf a} \neq {\bf b}} \{d({\bf a}, {\bf b}),  {\bf a} \in {\bf C}, {\bf b} \in {\bf C} \}.$$  One of the main goals in the theory of error-correcting codes is to construct codes ${\bf C} \subset {\bf F}_q^n$ with large cardinalities and minimum distances. In general, there are some upper bounds on cardinalities or minimum distances of codes.  Optimal codes attaining these bounds are particularly interesting, see \cite{MScode}. The Singleton bound for general $(n,M,d)_q$ codes asserts
$M \leq q^{n-d+1}$ and codes attaining this bound is called (nonlinear) maximal distance separable
(MDS) codes. Reed-Solomon codes are well-known linear MDS codes, see \cite{MScode}. Constructions of
non-Reed-Solomon type linear MDS codes have attracted many attentions, see, e.g.,  \cite{Chen1} and references therein. For counting MDS codes, we refer to \cite{Ghorpade}.\\

On the other hand, the diameter of a code ${\bf C} \subset {\bf F}_q^n$ is $$\delta({\bf C})=\max_{{\bf a} \neq {\bf b}} \{d({\bf a}, {\bf b}),  {\bf a} \in {\bf C}, {\bf b} \in {\bf C} \}.$$ It is also interesting to construct large anticodes with a fixed diameter, see \cite{AK,EKR,Kleitman,Farrell}. The following code-anticode bound in \cite{Delsarte,AAK} asserts that for a code ${\bf C} \subset {\bf F}_q^n$ with the minimum distance $d$ and an anticode ${\bf A} \subset {\bf F}_q^n$ with the diameter $d-1$, we have $$|{\bf C}| \cdot |{\bf A}| \leq q^n.$$

A linear code ${\bf C} \subset {\bf F}_q^n$ of the dimension $k$ and the minimum distance $d$ is denoted as a linear $[n,k,d]_q$ code.  It is clear that for a linear code, its minimum distance is just the minimum weight of nonzero codewords. Its diameter is the maximum weight of nonzero codewords. A linear code is called projective if the dual distance is at least three. In this paper, we only consider projective linear anticodes. For fixed $q,n,k$, if there is an linear $[n,k,d]_q$ code, and there is no $[n,k,d+1]_q$ code. We call this code optimal. If there is a linear $[n,k,d+1]_q$ code, but there is no linear $[n, k, d+2]_q$ code, we call this code almost optimal. We treat maximum weights of projective linear codes in this paper. The binary simplex $[2^k-1, k, 2^{k-1}]_2$ code can be considered as an anticode with the diameter $\delta=2^{k-1}$. For a linear code ${\bf C} \subset {\bf F}_q^n$, we denote by $A_i({\bf C})$ the number of codewords with the weight $i$, $0 \leq i \leq n$.

\subsection{Griesmer bound}

The Griesmer bound was proved in \cite{Griesmer}. It asserts that, for a linear $[n, k, d]_q$ code, $$n \geq \Sigma_{i=0}^{k-1} \lceil \frac{d}{q^i} \rceil,$$
or see \cite[Theorem 2.7.4, page 81]{HP}. Set $g_q(k,d)=\Sigma_{i=0}^{k-1} \lceil \frac{d}{q^i} \rceil$. The Griesmer defect of the linear code ${\bf C}$ is $g({\bf C})=n-g_q(k,d)$. A linear $[n, k, d]_q$ code attaining this bound, that is, the Griesmer defect zero code, is called a Griesmer code, see \cite{Solomon,Hu,Hu1}. A Griesmer defect one linear code is called an almost Griesmer code. It was proved in \cite{BM} that for fixed $q$ and $k$, there are Griesmer codes for any given fixed large minimum distance $d$. Therefore constructions of Griesmer codes have attracted many authors, see \cite{Solomon,Ding,Sihem2,Hu,Hu1,Xu,LL}. In \cite{Solomon}, many binary linear Griesmer codes were constructed and these codes are called Solomon-Stiffler codes later. These Solomon-Stiffler codes motivated the construction of binary linear anticodes in \cite{Farrell}, or see \cite[pp. 547-556]{MScode}. \\

\subsection{Linear codes with few nonzero weights}

Weight distributions are important in theory and practice. It is closely related to error probability of decoding algorithms, see \cite[Chapter 2]{Lint}. The famous Assmus-Mattson theorem in \cite{Assmus,DingTang} gives a construction of good designs from linear codes. Many constructions of designs from linear codes, see \cite{DingTang}, depend on weight distributions of dual codes. Therefore it is important and challenging to determine even partial weight distributions $A_i({\bf C})$ for some weights $i$ in the range $d\leq i \leq n$. Exact weight distributions $A_d({\bf C}), \ldots A_n({\bf C})$ are only determined for very few linear codes, such as MDS codes, some BCH codes and some linear codes with few nonzero weights, see \cite[Chapter 11]{MScode} and \cite{CG,Ding,DD,DD1,HengYue,HengDing,Kasami,carlet,DingTang,Xu}. Linear codes with few nonzero weights, or few-weight linear codes, have
been studied by many authors, see \cite{CG,Ding,DD,DD1,HengYue,Kasami,HengDing,carlet,DingTang}. Notice that for many few-weight linear codes constructed in previous papers \cite{Kasami,HengYue,Ding,DD,HengDing,DingTang,Xu}, their dual codes are known and have minimum distances at least three. Hence these few-weight linear codes can be thought as projective linear anticodes.\\

\subsection{Minimal codes}

Minimal codewords were first introduced in \cite{Huang} for decoding linear block codes. Then minimal binary linear codes were studied in \cite{CL} in the name of linear interesting codes. In \cite{AB}, minimal codes and minimal codewords were introduced for general $q$-ary linear codes and studied systematically.  A nonzero codeword ${\bf c}$ of a linear $[n,k,d]_q$ code ${\bf C}$ is called minimal, if for any codeword ${\bf c}'$ satisfying $supp({\bf c}') \subset supp({\bf c})$, then there is a nonzero $\lambda \in {\bf F}_q$, such that, ${\bf c}'=\lambda {\bf c}$. A linear code ${\bf C}$ is called minimal if each nonzero codeword is minimal. In \cite{Alfarano} combinatorial and geometric properties of minimal codes were studied.\\

In \cite{AB}, the Ashikhmin-Barg criterion $\frac{d({\bf C})}{\delta({\bf C})} > \frac{q-1}{q}$ was proposed as a sufficient condition for a $q$-ary linear code to be minimal. However, this is certainly not a necessary condition for minimality of linear codes. In \cite{CH} the first family of minimal binary linear codes violating the Ashikhmin-Barg criterion was constructed. Then several families of minimal binary linear codes violating the Ashikhmin-Barg criterion were constructed in \cite{DHZ} from binary linear codes with three nonzero weights, also see \cite{Alon,Sihem,Sihem2,Yan}.\\

\subsection{$l$-strongly walk-regular graphs}

We recall some basic facts about the $l$-strongly walk-regular graph (SWRG or $l$-SWRG).
A strongly regular graph (SRG) is a type of regular graph in which the number of common neighbors of two distinct vertices is determined only by their adjacency.
And $l$-SWRGs are the natural generalization of SRGs where length two paths are replaced by paths of length $l\geq2$, see \cite{Dam}.
The definition of an $l$-SWRG (defined as $\Gamma$) with parameters $(\lambda_{l}, \mu_{l}, \nu_{l})$ is that there are $\lambda_{l}$, $\mu_{l}$ and $\nu_{l}$ walks of length $l$ between every two adjacent, every two non-adjacent, and every two identical vertices, respectively.\\

SWRGs have wide applications in finite groups, codes, and design, see \cite{Dam,KKSS}. In \cite{KKSS}, it was proved that a binary three-weight projective linear code satisfying some conditions leads to a $3$-strongly walk-regular graph. In a recent paper \cite{ML},  an infinite family of binary three-weight projective linear codes with new parameters was constructed, which produces $l$-SWRGs for every odd $l \geq 3$, see \cite[Section 5]{ML}.\\

\subsection{Related works}

The intuition behind code-anticode bound is that all sets ${\bf C}+a$, $a \in {\bf A}$, are disjoint in ${\bf F}_q^n$. Sphere-packing bound can be thought as a special case of the code-anticode bound, see \cite{Delsarte,AAK,Byrne}. For the diametric theorem on $q$-ary anticodes, we refer to \cite{AK,AAK}. The Erd\"{o}s-Kleitman bound was proved in \cite{EKR} and \cite{Kleitman}, when Kleitman was with the Department of Physics in Brandies University. It asserts that an anticode in ${\bf F}_2^n$ with the diameter $\delta$ has at most $\Sigma_{i=0}^{\frac{\delta}{2}} \displaystyle{n \choose i}$ elements. It is clear that when $\delta=2$, the anticode consisting of $n+1$ elements ${\bf 0}, {\bf e}_1, \ldots, {\bf e}_n$ is an anticode attaining the Erd\"{o}s-Kleitman bound. Here ${\bf e}_i \in {\bf F}_2^n$ is the vector with the nonzero coordinate at the position $i$.\\

In \cite{Farrell}, the lower bound $$\delta\geq \frac{2^{k-1}n}{2^k-1}$$ for a binary linear projective anticode of the dimension $k$ and the diameter $\delta$ was proved. A Gilbert-like bound on anticode was proved in \cite{Reddy}.  From a similar idea of constructing Solomon-Stiffler codes, Farrell gave a construction of linear codes with optimal parameters or near optimal parameters from linear anticodes, see \cite{Farrell} and \cite[pp. 547-556]{MScode}, and vice versa. Optimal locally repairable codes were constructed from anticodes in \cite{SZ}. In \cite[Conclusion]{Zhou}, it was mentioned that the codes generated by the matrices deleting some columns in the generator matrices of the binary simplex codes, could be attractive. Similar idea motivated the paper \cite{LDT}. However the idea in \cite{Farrell,Solomon,Zhou,LDT} was not fully developed. We show that this relation between projective linear codes and anticodes can lead to numerous minimal linear codes and linear codes with two, three, and four nonzero weights, from our complementary theorem Corollary \ref{C-3-1}. More interestingly, some of these minimal linear codes and linear codes with few nonzero weights are optimal, almost optimal and near optimal. The basic point is as follows. Considering few-weight linear codes constructed in \cite{Ding,DD,DD1,HengYue,HengDing,Kasami,carlet,DingTang,Xu} as projective anticodes, since their maximum weights were determined, then complementary few-weight linear codes can be constructed immediately. The only remaining point is antiGrismer defect of these anticodes should be calculated. If the antiGriesmer defect is small, then the complementary code has a small Griesmer defect, thus optimal, almost optimal or near optimal.\\

It is well-known that a binary three-weight projective linear code with three weights $w_1<w_2<w_3$ and the length $n$ satisfying $$w_1+w_2+w_3=\frac{3n}{2}$$ lead to a $3$-strongly walk-regular graphs, see \cite{Dam,KKSS,ML}. Moreover if $w_2=\frac{n}{2}$, the binary three-weight projective linear code leads to $l$-strongly walk-regular graphs with each odd $l\geq 3$, see \cite{Dam,KKSS,ML}. Since several such binary three-weight projective linear codes are constructed in Section 5. They lead to many such $3$-strongly walk-regular and $l$-strongly walk-regular graphs, see Section 9.\\

\subsection{Our contributions}

Our main results are as follows. \\

\begin{theorem}[AntiGriesmer bound]\label{T-1-1}
Let $q$ be a prime power and $n$ be a positive integer satisfying $n<q^{k-1}$. Let ${\bf C} \subset {\bf F}_q^n$ be a projective linear code of the dimension $k$ and the maximum weight (diameter) $\delta$. Then
$$\Sigma_{i=0} \lfloor \frac{\delta}{q^i} \rfloor \geq n.$$
\end{theorem}

First of all, this bound can be thought as lower bounds on maximum weights of projective linear codes. This bound is stronger than the Erd\"{o}s-Kleitman bound. We construct many linear projective anticodes from fixed weight vectors, dual of BCH codes and linear codes with few nonzero weights, such that they are close this bound. From the proof of Theorem 2.1, if the complementary code attains the Griesmer bound, then the code attains the above antiGriesmer bound for projective linear anticodes.\\

\begin{theorem}\label{T-1-2}
Let $q$ be a prime power and $n$ be a positive integer satisfying $n<q^{k-1}$. Let ${\bf C} \subset {\bf F}_q^n$ be a projective linear code with the dimension $k_1$ and the diameter $\delta$. Let $k$ be a positive integer satisfying $k \geq k_1$ be a positive integer. Then a linear projective $[\frac{q^k-1}{q-1}-n, k, q^{k-1}-\delta]_q$ code is constructed.
\end{theorem}

\begin{theorem}[Complementary theorem]\label{T-1-3}
Let $q$ be a prime power and $n$ be a positive integer satisfying $n<q^{k-1}$. Let ${\bf C} \subset {\bf F}_q^n$ be a projective linear code with the dimension $k$ and $t$ nonzero weights $w_1<\cdots<w_t$. Then we construct an explicit $[\frac{q^k-1}{q-1}-n, k, q^{k-1}-\delta({\bf C})]_q$ code ${\bf C}'$ with $t$ nonzero weights $q^{k-1}-w_t<\cdots<q^{k-1}-w_1$. Moreover $$A_{w_j}({\bf C})=A_{q^{k-1}-w_j}({\bf C}').$$ Let $K$ be a positive integer satisfying $K>k$, then we construct an explicit $[\frac{q^K-1}{q-1},K, q^{K-1}-\delta({\bf C})]_q$ code ${\bf C}''$ with $t+1$ weights $q^{K-1}-w_t<\cdots<q^{K-1}-w_1<q^{K-1}$. Moreover $$q^{K-k}A_{w_j}({\bf C})=A_{q^{K-1}-w_j}({\bf C}''),$$ and $$A_{q^{K-1}}({\bf C}'')=q^{K-k}-1.$$
\end{theorem}

Similarly, we introduce the antiGriesmer defect of a linear code as $$\Sigma_{i=0}^{k-1} \lfloor \frac{\delta({\bf C})}{q^i} \rfloor-n.$$ This defect can reflect the Griesmer defect of its complementary code. Projective linear anticode meeting the above bound are called antiGriesmer codes. From the proof of Theorem 2.1, the complementary code of a projective Griesmer code is an antiGriesmer code, ans vice versa. In general, complementary codes of small antiGriesmer defect anticodes have small Griesmer defect. Therefore, we can consider these few-weight linear codes constructed in \cite{Ding,DD,DD1,HengYue,Kasami,HengDing,carlet,DingTang,Sihem,Sihem2,Xu} as projective linear anticodes, and calculate their antiGriesmer defects. If their antiGriesmer defects are small, complementary codes are close to the Griemser bound.\\

From this complementary theorem, infinitely many projective linear codes with $t+1$ nonzero weights close to the optimal parameters are constructed, see Section 5 below. Notice that these codes in the following table are minimal binary linear codes, since the Ashikhmin-Barg criterion is satisfied. As far as we know, no such anticode-based construction of few-weight linear codes has been reported in the literature. In the following Tables 1 and 2, we list some optimal, almost optimal, or near optimal few-weight linear codes constructed in this paper.\\

\begin{longtable}{|l|l|l|l|}
\caption{\label{tab-1-1} Minimal two-weight and four-weight codes}\\ \hline    
$q$ &Parameters     &Weights                &Optimal parameters                   \\ \hline
$2$ &$[9,4,4]_2$    &$4,5,6,8$              &optimal                              \\ \hline
$2$ &$[10,4,4]_2$   &$4,6$                  &optimal                              \\ \hline
$2$ &$[11,4,5]_2$   &$5,6,7,8$              &optimal                              \\ \hline
$2$ &$[23,5,10]_2$  &$10,12,14,16$          &almost optimal                       \\ \hline
$2$ &$[25,5,12]_2$  &$12,13,14,16$          &optimal                              \\ \hline
$2$ &$[27,5,13]_2$  &$13,14,15,16$          &optimal                              \\ \hline
$2$ &$[30,8,8]_2$   &$8,16$                 &$[30,8,12]_2$                        \\ \hline
$2$ &$[35,6,16]_2$  &$16,20$                &optimal                              \\ \hline
$2$ &$[43,6,20]_2$  &$20,22,24,32$          &optimal                              \\ \hline
$2$ &$[46,6,20]_2$  &$20,21,22,23,24,25,26$ &optimal                              \\ \hline
$2$ &$[47,6,22]_2$  &$22,24,26,32$          &almost optimal                       \\ \hline
$2$ &$[48,6,22]_2$  &$22,24,26,32$          &$[48,6,24]_2$                        \\ \hline
$2$ &$[51,6,24]_2$  &$24,26,28,32$          &optimal                              \\ \hline
$2$ &$[51,8,24]_2$  &$24,32$                &optimal                              \\ \hline
$2$ &$[53,6,22]_2$  &$22,26,28,32$          &$[53,6,26]_2$                        \\ \hline
$2$ &$[55,6,26]_2$  &$26,28,30,32$          &almost optimal                       \\ \hline
$2$ &$[57,6,28]_2$  &$28,29,30,32$          &optimal                              \\ \hline
$2$ &$[59,6,29]_2$  &$29,30,31,32$          &optimal                              \\ \hline
$2$ &$[88,7,40]_2$  &$40,44,48,64$          &$[88,7,42]_2$                        \\ \hline
$2$ &$[95,7,44]_2$  &$44,48,52,64$          &$[95,7,47]_2$                        \\ \hline
$2$ &$[96,7,44]_2$  &$44,48,52,64$          &$[96,7,48]_2$                        \\ \hline
$2$ &$[107,7,52]_2$ &$52,54,56,64$          &optimal                              \\ \hline
$2$ &$[111,7,54]_2$ &$54,56,58,64$          &almost optimal                       \\ \hline
$2$ &$[112,7,54]_2$ &$54,56,58,64$          &$[112,7,56]_2$                       \\ \hline
$2$ &$[115,7,56]_2$ &$56,58,60,64$          &optimal                              \\ \hline
$2$ &$[117,7,54]_2$ &$54,58,60,64$          &$[117,7,58]_2$                       \\ \hline
$2$ &$[119,7,58]_2$ &$58,60,62,64$          &almost optimal                       \\ \hline
$2$ &$[120,7,58]_2$ &$58,60,62,64$          &$[120,7,60]_2$                       \\ \hline
$2$ &$[121,7,60]_2$ &$60,61,62,64$          &optimal                              \\ \hline
$2$ &$[123,7,61]_2$ &$61,62,63,64$          &optimal                              \\ \hline
$2$ &$[126,8,56]_2$ &$56,60,66,70$          &$[126,8,62]_2$                       \\ \hline
$2$ &$[183,8,88]_2$ &$88,92,96,128$         &$[183,8,90]_2$                       \\ \hline
$2$ &$[191,8,92]_2$ &$92,96,100,128$        &$[191,8,95]_2$                       \\ \hline
$2$ &$[192,8,92]_2$ &$92,96,100,128$        &$[192,8,96]_2$                       \\ \hline
$2$ &$[199,8,96]_2$ &$96,100,104,128$       &$[199,8,98]_2$                       \\ \hline
$2$ &$[204,8,96]_2$ &$96,104$               &$[204,8,100]_2$                      \\ \hline
$2$ &$[216,8,104]_2$&$104,108,112,128$      &$[216,8,107]_2$                      \\ \hline
$2$ &$[223,8,108]_2$&$108,112,116,128$      &$[223,8,111]_2$                      \\ \hline
$2$ &$[224,8,108]_2$&$108,112,116,128$      &$[224,8,112]_2$                      \\ \hline
$2$ &$[225,8,112]_2$&$112,120$              &optimal                              \\ \hline
$2$ &$[235,8,116]_2$&$116,118,120,128$      &optimal                              \\ \hline
$2$ &$[239,8,118]_2$&$118,120,122,128$      &almost optimal                       \\ \hline
$2$ &$[240,8,118]_2$&$118,120,122,128$      &$[240,8,120]_2$                      \\ \hline
$2$ &$[245,8,118]_2$&$118,122,124,128$      &$[245,8,122]_2$                      \\ \hline
$2$ &$[243,8,120]_2$&$120,122,124,128$      &optimal                              \\ \hline
$2$ &$[247,8,122]_2$&$122,124,126,128$      &almost optimal                       \\ \hline
$2$ &$[248,8,122]_2$&$122,124,126,128$      &$[248,8,124]_2$                      \\ \hline
$2$ &$[249,8,124]_2$&$124,125,126,128$      &optimal                              \\ \hline
$2$ &$[251,8,125]_2$&$125,126,127,128$      &optimal                              \\ \hline
$3$ &$[30,4,18]_3$  &$18,21$                &almost optimal                       \\ \hline
$3$ &$[32,4,21]_3$  &$21,24$                &optimal                              \\ \hline
$4$ &$[68,4, 48]_4$ &$48,52$                &$[68,4,50]_4$                        \\ \hline
$4$ &$[75,4,56]_4$  &$56,69$                & optimal                             \\ \hline
$5$ &$[130,4,100]_5$&$100,105$              &$[130,4,103]_5$                      \\ \hline
\end{longtable}

\begin{longtable}{|l|l|l|l|}
\caption{\label{tab-1-3} Minimal three-weight codes}\\ \hline
$q$ &Parameters     &Weights&Optimal parameters     \\ \hline
$2$ &$[16,5,6]_2$   &$6,8,10$       &$[16,5,8]_2$   \\ \hline
$2$ &$[19,5,8]_2$   &$8,10,12$      &optimal        \\ \hline
$2$ &$[21,5,6]_2$   &$6,10,12$      &$[21,5,10]_2$   \\ \hline
$2$ &$[22,5,10]_2$  &$10,12,16$     &optimal        \\ \hline
$2$ &$[25,5,10]_2$  &$10,12,14$     &$[25,5,12]_2$  \\ \hline
$2$ &$[32,6,12]_2$  &$12,16,20$     &$[32,6,16]_2$  \\ \hline
$2$ &$[48,6,22]_2$  &$22,24,26$     &$[46,6,24]_2$  \\ \hline
$2$ &$[54,6,26]_2$  &$26,28,32$     &optimal        \\ \hline
$2$ &$[56,6,26]_2$  &$26,28,30$     &$[56,6,28]_2$  \\ \hline
$2$ &$[64,7,28]_2$  &$28,32,36$     &$[64,7,32]_2$  \\ \hline
$2$ &$[70,7,32]_2$  &$32,35,40$     &almost optimal \\ \hline
$2$ &$[71,7,32]_2$  &$32,36,40$     &$[71,7,34]_2$  \\ \hline
$2$ &$[92,7,44]_2$  &$44,48,64$     &almost optimal \\ \hline
$2$ &$[118,7,58]_2$ &$58,60,64$     &optimal        \\ \hline
$2$ &$[144,8,64]_2$ &$64,72,80$     &$[144,8,70]_2$ \\ \hline
$2$ &$[220,8,108]_2$&$108,112,128$  &almost optimal \\ \hline
$2$ &$[246,8,122]_2$&$122,124,128$  &optimal        \\ \hline
$3$ &$[81,5,51]_3$  &$51,54,57$     &$[81,5,54]_3$  \\ \hline
$4$ &$[18,3,13]_4$  &$13,14,15$     &optimal        \\ \hline
$5$ &$[25,3,19]_5$  &$19,20,21$     &almost optimal \\ \hline
$5$ &$[28,3,22]_5$  &$22,23,24$     &optimal        \\ \hline
$7$ &$[49,3,41]_7$  &$41,42,43$     &almost optimal \\ \hline
$7$ &$[54,3,46]_7$  &$46,47,48$     &optimal        \\ \hline
$8$ &$[70,3,61]_8$  &$61,62,63$     &optimal        \\ \hline
$9$ &$[81,3,71]_9$  &$71,72,73$     &almost optimal \\ \hline
$9$ &$[88,3,78]_9$  &$78,79,80$     &optimal        \\ \hline
\end{longtable}

Many minimal binary linear codes satisfying the Ashikhmin-Barg criterion are constructed from the above complementary theorem, see Section 7 below. In the following table some minimal linear codes with optimal or near optimal parameters are listed.\\

\begin{longtable}{|l|l|l|l|l|}
\caption{\label{tab-1-4} Minimal codes}\\ \hline
$q$&Length&Dimension&Distance&Optimal parameters\\ \hline
$2$ &$16$& $5$ &$7$  &almost optimal\\ \hline
$2$ &$28$& $6$ &$12$ &optimal\\ \hline
$2$ &$35$& $6$ &$16$ &optimal\\ \hline
$2$ &$42$& $6$ &$20$ &optimal \\ \hline
$2$ &$48$& $6$ &$23$ &almost optimal \\ \hline
$2$ &$57$& $7$ &$24$ &$[57,7,26]_2$ \\ \hline
$2$ &$70$& $7$ & $32$ &almost optimal \\ \hline
$2$ &$92$& $7$ & $44$ &almost optimal \\ \hline
$2$ &$99$& $7$ &$48$ &optimal \\ \hline
$2$ &$106$& $7$ &$52$ &optimal \\ \hline
$2$ &$112$& $7$ &$55$ &almost optimal \\ \hline
$2$ &$185$& $8$ &$88$ &$[185,8,91]_2$ \\ \hline
$2$ &$220$& $8$ &$108$ &almost optimal \\ \hline
$2$ &$227$& $8$ &$112$ &optimal \\ \hline
$2$ &$234$& $8$ &$116$ &optimal \\ \hline
$2$ &$240$& $8$ &$119$ &almost optimal \\ \hline
\end{longtable}

We also construct complementary Reed-Solomon codes as in the following result, see Section 8.\\

\begin{theorem}[Complementary Reed-Solomon codes]\label{T-1-4}
A complementary Reed-Solomon $[\frac{q^k-1}{q-1}-q,k,q^{k-1}-q]_q$ code is a minimal $k$-weight almost Griesmer code. As an anticode, it has a difference $k-1$ to the antiGriesmer bound.
\end{theorem}

Complementary Reed-Solomon codes over small fields are optimal minimal codes, as listed in the following table.\\

\begin{longtable}{|l|l|l|l|}
\caption{\label{tab-1-5} Minimal complementary RS codes over small fields}\\ \hline
$q$&Parameters&Weights&Optimality\\ \hline
$2$ &$[29,5,14]_2$&$14,15, 16,17,18$&optimal\\ \hline
$2$ &$[61,6,30]_2$&$30,31,32,33,34,35$&optimal\\ \hline
$2$ &$[125,7,62]_2$&$62, 63,64,65,66,67,68$&optimal\\ \hline
$2$ &$[253,8,126]_2$&$126,127,128,129,130,131,132,133$&optimal\\ \hline
$3$ &$[37,4,24]_3$&$24,25,26,27$&optimal\\ \hline
$3$ &$[118,5,78]_3$&$78,79,80,81,82$&optimal\\ \hline
$4$ &$[17,3,12]_4$& $12,13,14$& optimal \\ \hline
$4$ &$[81,4,60]_4$&$60,61,62,63$& optimal \\ \hline
$5$ &$[26,3,20]_5$&$20,21,22$&optimal\\ \hline
$7$ &$[50,3,42]_7$&$42,43,44$&optimal\\ \hline
$8$ &$[65,3,56]_8$&$56,57,58$&optimal\\ \hline
$9$ &$[82,3,72]_9$&$72,73,74$&optimal\\ \hline
\end{longtable}	

Notice that it was proved in \cite[Theorem 1 and 5]{Hell3}, some binary and $q$-ary Griesmer codes are actually Solomon-Stiffler codes or Belov codes. Therefore, some optimal codes meeting the Griesmer bound constructed in this paper are not new.\\

\section{Bounds for anticodes}\label{sec-2}

We first recall the following code-anticode bound due to \cite{Delsarte,AAK}, or see \cite{Byrne}.\\

{\bf Code-anticode bound.} {\em Let ${\bf C} \subset {\bf F}_q^n$ be a code with the minimum distance $d$ and ${\bf A} \subset {\bf F}_q^n$ be an anticode with the diameter $d-1$. Then $$|{\bf C}| \cdot |{\bf A}| \leq q^n.$$}

\begin{proof}{\rm
We consider all sets of the ${\bf C}+a$, where $a$ takes all elements of the anticode ${\bf A}$. Then these sets are disjoint subsets of ${\bf F}_q^n$. Otherwise, there are two elements $a$ and $a'$ of the anticode ${\bf A}$ such that there is an common element $c+a$ and $c'+a'$ in ${\bf C}+a$ and ${\bf C}+a'$.  Therefore $$c-c'=a'-a,$$ we have $$d(a,a')=d(c,c') \geq d.$$ This is a contradiction.\\

Let us observe the implication of code-anticode bound in the Hamming metric space ${\bf F}_2^n$. From the diameter two anticode with $n+1$ codewords described in Subsection 1.6, we have the upper bound $$M \leq \frac{2^n}{n+1},$$ for a binary $(n,M,3)_2$ code. Binary Hamming $[2^m-1, 2^m-m-1,3]_2$ code is certainly such a code attaining this upper bound.}
\end{proof}

Let ${\bf C}\subset {\bf F}_q^n$ be a projective linear code with one generator matrix ${\bf G}$ consisting of $n$ distinct columns ${\bf g}_1, \ldots, {\bf g}_n$ in ${\bf F}_q^k$. Let ${\bf G}$ also be the set of these $n$ vectors in ${\bf F}_q^k$. Let $H \subset {\bf F}_q^k$ be an arbitrary dimension $k-1$ subspace, it is clear that the maximum weight of ${\bf C}$ is $$\delta({\bf C})=n-\min_{H \subset {\bf F}_q^k, \dim H=k-1} |H \bigcap {\bf G}|.$$ and the minimum weight  of ${\bf C}$ is $$d({\bf C})=n-\max_{H \subset {\bf F}_q^k, \dim H=k-1}|H \bigcap {\bf G}|.$$ \\
The following proposition is clear.\\

\begin{proposition}\label{P-2-1}
Let ${\bf C} \subset {\bf F}_q^n$ be a linear $[n,k]_q$ code. Then its maximum weight is at least $k$.
\end{proposition}

\begin{proof}{\rm
There is a nonzero codeword with $k$ nonzero coordinates.}
\end{proof}

The linear simplex $[3,2,2]_2$ code is certainly a linear projective anticode attaining this bound.\\

We introduce the complementary code of a projective linear $[n,k]_q$ anticode ${\bf C}$ as follows. Here the length satisfies the condition $n<q^{k-1}$.  Let ${\bf G}$ be the $k \times n$ generator matrix of the projective linear anticode ${\bf C}$. Then $n$ columns ${\bf g}_1, \ldots, {\bf g}_n$ of ${\bf G}$ are distinct vectors of ${\bf F}_q^k$. By deleting these $n$ distinct columns in the generator matrix of the $q$-ary simplex $[\frac{q^k-1}{q-1},k, q^{k-1}]_q$ code, we get a linear $[\frac{q^k-1}{q-1}-n, k, q^{k-1}-\delta({\bf C})]_q$ code. The dimension is $k$, because there are $\frac{q^k-1}{q-1}-n>\frac{q^{k-1}-1}{q-1}$ columns in the generator matrix of this code. These columns cannot be located in a $k-1$ dimension linear subspace.\\

For example, we consider the simplex $[2^{k-1}-1,k-1,2^{k-2}]$ as a projective linear anticode with the diameter $\delta({\bf C})=d({\bf C})$, the complementary code is actually a first order Reed-Muller $[2^{k-1},k,2^{k-2}]_2$ code.\\

We prove the following antiGriesmer lower bound for projective linear anticode.\\

\begin{theorem}\label{T-2-1}
Let $q$ be a prime power and $n$ be a positive integer satisfying $n<q^{k-1}$. Let ${\bf C} \subset {\bf F}_q^n$ be a projective linear code of the dimension $k$. Then its maximum weight $\delta({\bf C})$ satisfies $$\Sigma_{i=0}^{k-1} \lfloor \frac{\delta({\bf C})}{q^i} \rfloor \geq n.$$
\end{theorem}

\begin{proof}{\rm
From the Griesmer bound for the complementary code, we have $$\frac{q^k-1}{q-1}-n \geq \Sigma_{i=0}^{k-1} \lceil \frac{q^{k-1}-\delta({\bf C})}{q^i}\rceil. $$ On the other hand, $$\lceil \frac{q^{k-1}-\delta({\bf C})}{q^i}\rceil=q^{k-1-i}-\lfloor \frac{\delta({\bf C})}{q^i} \rfloor.$$ The conclusion follows immediately.\\}
\end{proof}

The following lower bound for projective linear anticode can be deduced from antiGriesmer bound directly.\\

\begin{corollary}\label{C-2-1}
Let $q$ be a prime power and $n$ be a positive integer satisfying $n<q^{k-1}$. Let ${\bf C} \subset {\bf F}_q^n$ be a projective linear code of the dimension $k$. Then its maximum weight (diameter) is at least $(1-\frac{1}{q})n$.\\
\end{corollary}

\begin{proof}{\rm
It is clear that $$\delta({\bf C})(1+\frac{1}{q}+\cdots+\frac{1}{q^{k-1}}) \geq \Sigma_{i=0}^{k-1} \lfloor \frac{\delta({\bf C})}{q^i} \rfloor \geq n.$$ The conclusion follows immediately.}
\end{proof}

The above two bounds are stronger than the classical Erd\"{o}s-Kleitman bound for general binary anticodes. Since when $\delta({\bf C})=2$, there are optimal binary anticode attaining the Erd\"{o}s-Kleitman bound as described in Subsection 1.6. However, for a binary linear projective anticode with the dimension $k$ and the diameter two. If the length $n \leq 2^{k-1}$, then $n \leq 4$. Therefore the above two bounds give much stronger restrictions on binary projective linear anticodes than the Erd\"{o}s-Kleitman bound.\\

\section{Linear codes from projective linear anticodes}\label{sec-3}

In this section, we prove the main results Theorem 1.3 and 1.4.\\

{\bf Theorem 3.1} {\em Let $q$ be a prime power and $n$ be a positive integer satisfying $n<q^{k-1}$. Let ${\bf C} \subset {\bf F}_q^n$ be a projective linear code with the dimension $k_1$ and the maximum weight $\delta$. Let $k \geq k_1$ be a positive integer. Then a projective linear $[\frac{q^k-1}{q-1}-n, k, q^{k-1}-\delta]_q$ code is constructed.}\\

{\bf Proof.} Let ${\bf g}_1 \ldots, {\bf g}_n$ be $n$ distinct columns of the projective linear code ${\bf C}$. By pending some zero coordinates, we can consider them as $n$ distinct vectors in ${\bf F}_q^k$. By deleting these columns in $\frac{q^k-1}{q-1}$ columns of the $q$-ary simplex code, we get a projective $[\frac{q^k-1}{q-1}-n, k, q^{k-1}-\delta]_q$ code.\\

It is clear that from the interpretation of a projective linear code, the weight of a codeword is $n-|H \bigcap {\bf G}|$, then we have the following result immediately.\\

\begin{corollary}[Complementary theorem] \label{C-3-1}
Let $q$ be a prime power and $n$ be a positive integer satisfying $n<q^{k-1}$. Let ${\bf C} \subset {\bf F}_q^n$ be a projective linear code with the dimension $k$ and $t$ nonzero weights $w_1<\cdots<w_t$. Then we construct an explicit $[\frac{q^k-1}{q-1}-n, k, q^{k-1}-\delta({\bf C})]_q$ code ${\bf C}'$ with $t$ nonzero weights $q^{k-1}-w_t<\cdots<q^{k-1}-w_1$. Moreover $$A_{w_j}({\bf C})=A_{q^{k-1}-w_j}({\bf C}').$$  Let $K$ be a positive integer satisfying $K>k$, then we construct an explicit $[\frac{q^K-1}{q-},K, q^{K-1}-\delta({\bf C})]_q$ code ${\bf C}''$ with $t+1$ weights $q^{K-1}-w_t<\cdots<q^{K-1}-w_1<q^{K-1}$. Moreover $$q^{K-k}A_{w_j}({\bf C})=A_{q^{K-1}-w_j}({\bf C}''),$$ and $$A_{q^{K-1}}({\bf C}'')=q^{K-k}-1.$$
\end{corollary}

{\bf Proof.} When $K=k$, the conclusion follows directly. When $K>k$. We delete columns ${\bf g}_1',\ldots, {\bf g}_n'$
in the generator matrix of the $q$-ary simplex $[\frac{q^K-1}{q-1},K,q^{K-1}]_q$ code. Then it is the generator matrix ${\bf G}'$ of the complementary code.  Here ${\bf g}_i'=({\bf 0},{\bf g}_i)$, $i=1,\ldots,n$, are vectors in ${\bf F}_q^K$, with $K-k$ zero coordinates. It is clear that for nonzero vectors of the form ${\bf X}=({\bf x}, {\bf 0})$, where ${\bf x} \in {\bf F}_q^{K-k}$, there are $q^{K-k}-1$ weight $q^{K-1}$ codewords ${\bf X} \cdot {\bf G}'$ in the complementary code. Considering the the first $K-k$ zero coordinates of ${\bf g}_1', \ldots, {\bf g}_n'$, the conclusion $$q^{K-k}A_{w_j}({\bf C})=A_{q^{K-1}-w_j}({\bf C}'')$$ follows immediately.\\

\section{Constructions of projective linear anticodes}\label{sec-4}

In this section, we give several constructions of projective linear anticodes close to the antiGriesmer bound.\\

\subsection{$t$-weight linear anticodes}

Let $n=2^m-1$, $t$ be a positive integer satisfying $$2t-1 <2^{\lceil \frac{m}{2} \rceil}+1.$$ Then in the dual code of the binary BCH code with the designed distance $2t+1$, then the weight $w$ of any nonzero codeword is in the range $$2^{m-1}-(t-1)2^{\frac{m}{2}}\leq w \leq 2^{m-1}+(t-1)2^{\frac{m}{2}}.$$ This is the famous Carlitz-Uchiyama theorem, see \cite[Page 280]{MScode}. Therefore we have projective linear $[2^m-1, 2m, 2^{m-1}-2^{\frac{m}{2}}]_2$ codes with the maximum weight (diameter) $2^{m-1}+2^{\frac{m}{2}}$, when $m$ is an even positive integer. It is clear that these projective linear codes have their maximum weights close to the Plotkin type bound.\\

Let $q=p^e$ be an odd prime power and $k$ be an odd positive integer. Then from \cite[Theorem 15]{HengYue}, projective linear four-weight $[q^k-1, k+1, q^{k-1}(q-1)-1-\frac{\sqrt{p^*}^{e(k+1)}}{q}]_q$ code was constructed. Four weights are $
 q^{k-1}(q-1)-1-\frac{\sqrt{p^*}^{e(k+1)}}{q}$, $q^{k-1}(q-1)$, $q^{k-1}(q-1)-1+\frac{\sqrt{p^*}^{e(k+1)}}{q}$,  and $q^k-1$, where $p^*=(-1)^{\frac{p-1}{2}}p$.\\

On the other hand, from our complementary theorem, we have a projective linear four-weight $[q^{k-1}+q^{k-2}+\cdots+q+2, k+1, 1]_q$ code with four weights $1$, $q^{k-1}+1-\frac{\sqrt{p^*}^{e(k+1)}}{q}$, $q^{k-1}$, $q^{k-1}+1+\frac{\sqrt{p^*}^{e(k+1)}}{q}$. Then the maximum weight is $q^{k-1}+1+\frac{\sqrt{p^*}^{e(k+1)}}{q}$, which is larger than and close to $\frac{q-1}{q} \cdot (q^{k-1}+q^{k-2}+\cdots+q+2)$.\\

The following table lists some projective linear anticodes obtained from projective linear codes obtained in previous papers \cite{DD1,DD,Kasami}\\

\begin{longtable}{|l|l|l|l|l|}
\caption{\label{tab:A-q-5-3} Projective linear anticodes}\\ \hline
$q$&Parameters&Maximum weight (Diameter) &$\frac{\delta}{n}$&References\\ \hline
$2$ &$[2^m-1,2m,2^{m-1}-2^{m/2}]_2$&$2^{m-1}+2^{m/2}$&$>\frac{1}{2}+\frac{1}{2^{m/2}}$&[23]\\ \hline
$2$ &$[2^{2m}-1,3m,2^{2m-1}-2^{m-1}]_2$&$2^{m-1}+2^{m-1}$&$>\frac{1}{2}+\frac{1}{2^{m+1}}$&[29]\\ \hline
$2$ &$[2^{m-1}-1,m,2^{m-2}-2^{\frac{m+l-4}{2}}]_2$&$2^{m-2}+2^{\frac{m+l-4}{2}}$&$>\frac{1}{2}+\frac{1}{2^{m/2}}$&[19]\\ \hline
$3$ &$[\frac{3^m-1}{2},2m,3^{m-1}-3^{\frac{m-1}{2}}]_3$&$3^{m-1}+3^{\frac{m-1}{2}}$&$>\frac{2}{3}+\frac{2}{3^{m/2}}$&[23]\\ \hline
$q$ &$[\frac{q^k-1}{q-1},k,q^{k-1}-1]_q$&$q^{k-1}$&$\frac{q^{k-2}}{q^{k-2}+\cdots+q+1}$&[48]\\ \hline
$q$ &$[q^2+1,4,q^2-q]_q$&$q^2$&$\frac{q^2}{q^2+1}$&[23]\\ \hline
$q$ &$[2q+2,4,q]_q$&$2q$&$\frac{q}{q+1}$AntiGriesmer&Section 5\\ \hline

\end{longtable}	

We refer to \cite{DingTang} and Section 5 for last two projective linear anticodes in the above table.\\

\subsection{Fixed weight binary vectors}

Some good binary linear anticodes were constructed in \cite[Theorem 3]{SZ} and applied to construct optimal locally reparable codes. In this section, we recover and extend their constructions of binary projective linear anticodes without using graph theory. Actually, this kind of binary projective linear codes were constructed in \cite{Zhou,Yan} as binary LCD codes, or minimal binary codes. The weight distribution was also calculated in \cite{Zhou,Yan}.\\

Let $k=2s$ be an even positive integer. Let ${\bf g}_1, \ldots, {\bf g}_n$, where $n=s(2s-1)$, be all length $k$ binary vectors with the fixed weight two. Then ${\bf G}=({\bf g}_1, \ldots, {\bf g}_n)$ is a $k \times n$ matrix with the rank $k-1$, since the sum of all coordinates in each column is zero. Then for a fixed subset $I \subset \{1,\ldots, k\}$ with $u$ elements, we consider the hyperplane $H_I$ defined by $$\Sigma_{i \in I} x_i =0.$$ It is clear that the number of columns in $H_I$ is $\frac{u(u-1)}{2}+\frac{(2s-u)(2s-u-1)}{2}$. Then this number is minimal when $u=s$. We construct a projective $[s(2s-1),2s-1]_2$ anticode with the maximum weight $s^2$.  When $k=2s+1$ is odd, from a similar argument, the maximum weight is $\frac{k^2-1}{4}$.\\

Similarly, we consider length $k$ vectors of fixed weight four. Set $k=7$, there are $35$ weight four vectors in ${\bf F}_2^7$. Then the matrix with these $35$ columns generator a projective linear $[35,6]_2$ code ${\bf C}$. Suppose that the hyperplane is $\Sigma_{i \in I} x_i=0$, where $I \subset \{1,2,\ldots,7\}$. We can consider the three possibilities $|I|=6$, $|I|=5$. and $|I|=4$. In the case $|I|=6$, the number of $35$ columns in this hyperplane is $$\frac{6 \cdot 5 \cdot 4 \cdot 3}{4 \cdot 3 \cdot 2\cdot 1}=15.$$ In the case $|I|=5$, the number of $35$ columns in this hyperplane is $$\frac{5 \cdot 4}{2\cdot 1}+5=15.$$ In the case $|I|=4$, the number of $35$ columns in this hyperplane is $$\frac{4 \cdot 3}{2\cdot 1} \cdot 3+1=19.$$ Then ${\bf C}$ is a two-weight linear $[35,6,16]_2$ code with two weights $16,20$. Comparing with the optimal $[35,6,16]_2$ code in \cite{Grassl}, it is a three-weight code with weights $16,20,24$. As a anticode, our code is better, since the diameter is smaller. This is an optimal two-weight $[35,6,16]_2$ code with the following weight distribution.\\

\begin{longtable}{|l|l|}
\caption{\label{tab:A-q-5-3} Weight distribution of $[35,6,16]_2$ code}\\ \hline
Weight&Weight distribution \\ \hline
$0$&$1$ \\ \hline
$16$&$35$ \\ \hline
$20$&$28$ \\ \hline
\end{longtable}	

From our complementary theorem, a three-weight linear $[92,7,44]_2$ code and a three weight linear $[220,8,108]_2$ are constructed from the second conclusion of Theorem 1.3. They are almost optimal codes.\\

Set $k=8$, there are $70$ weight four vectors in ${\bf F}_2^8$. Then the matrix with these $70$ columns generator a projective linear $[70,7]_2$ code ${\bf C}$. Suppose that the hyperplane is $\Sigma_{i \in I} x_i=0$, where $I \subset \{1,2,\ldots,8\}$. We can consider the four possibilities $|I|=7$, $|I|=6$. $|I|=5$ and $|I|=4$. In the case $|I|=7$, the number of $70$ columns in this hyperplane is $$\frac{7 \cdot 6 \cdot 5 \cdot 4}{4 \cdot 3 \cdot 2 \cdot 1}=35.$$ In the case $|I|=6$, the number of $70$ columns in this hyperplane is $$\frac{6 \cdot 5 \cdot 4 \cdot 3}{4 \cdot 3 \cdot 2 \cdot 1}+\frac{6 \cdot 5}{2 \cdot 1}=30.$$ In the case $|I|=5$, the number of $70$ columns in this hyperplane is $$\frac{5 \cdot 4 }{2 \cdot 1} \cdot 3+5=35.$$ In the case $|I|=4$, the number of $70$ columns in this hyperplane is $$\frac{4 \cdot 3 }{2 \cdot 1} \cdot \frac{4 \cdot 3 }{2 \cdot 1} +2=38.$$ Then ${\bf C}$ is a three-weight linear $[70,7,32]_2$ code with three weights $32,35,40$. Notice that the optimal minimum distance of a binary linear $[70,7]_2$ code is $33$, see \cite{Grassl}. The weight distribution of this three-weight linear $[70,7,32]_2$ code is as in the following table.\\

\begin{longtable}{|l|l|}
\caption{\label{tab:A-q-5-3} Weight distribution of the $[70,7,32]_2$ code}\\ \hline
Weight&Weight distribution \\ \hline
$0$&$1$ \\ \hline
$32$&$35$ \\ \hline
$35$&$64$ \\ \hline
$40$&$28$ \\ \hline
\end{longtable}

When $k=9$, we get a four-weight binary linear $[126,8,56]_2$ code with four weights $56,60,66,70$, the corresponding optimal minimum distance is $62$, see \cite{Grassl}. From Theorem 1.3, a four-weight binary linear $[129,8,58]_2$ code with four weights $58, 62, 68, 72$, is constructed, the corresponding optimal minimum distance is $64$, see \cite{Grassl}.\\

\section{Constructions of $t$-weight codes with optimal or near optimal parameters}\label{sec-5}

In this section, we construct many $t$-weight or $(t+1)$-weight projective linear codes from our complementary theorem and known $t$-weight codes constructed in \cite{Kasami,MScode,carlet,CG,DD1,DD,DingTang}. The basic point is as follows, if a $t$-weight projective linear code has its maximum weight close to the antiGriesmer bound, then its complementary $t$ or $(t+1)$-weight projective linear codes are close to the Griesmer bound. Therefore, it is necessary to check the antiGriesmer defect of constructed few-weight $q$-ary projective linear codes.\\

Let $m$ be an odd positive integer. The dual code of the primitive BCH $[2^m-1, 2^m-1-2m, 3]_2$
code is a linear $[2^m-1, 2m, 2^{m-1}-2^{\frac{m-1}{2}}]_2$ code with three nonzero weights $2^{m-1}-2^{\frac{m-1}{2}}, 2^{m-1}, 2^{m-1}+2^{\frac{m-1}{2}}$, see \cite[Theorem 4]{carlet}. These dual codes are obviously projective. Then we have a family of $[2^{2m}-2^m, 2m, 2^{2m-1}-2^{m-1}-2^{\frac{m-1}{2}}]_2$ codes with three nonzero weights from our complementary theorem Corollary \ref{C-3-1}. The three nonzero weights are $2^{2m-1}-2^{m-1}-2^{\frac{m-1}{2}}, 2^{2m-1}-2^{m-1}, 2^{2m-1}-2^{m-1}+2^{\frac{m-1}{2}}$. The weight distributions are
$$A_{2^{2m-1}-2^{m-1}-2^{\frac{m-1}{2}}}=(2^m-1)(2^{m-2}-2^{\frac{m-3}{2}}),$$
$$A_{2^{2m-1}-2^{m-1}}=(2^m-1)(2^{m-1}+1),$$
and
$$A_{2^{2m-1}-2^{m-1}+2^{\frac{m-1}{2}}}=(2^m-1)(2^{m-2}+2^{\frac{m-3}{2}}),$$
see \cite[Theorem 4]{carlet}. In particular, when $m=3$, this is a linear $[56,6,26]_2$ code with three nonzero weights $26, 28, 30$. Notice that the optimal minimum distance of linear $[56,6]_2$ code in \cite{Grassl} is $28$. This binary three-weight linear code have near optimal parameters. Actually, Griesmer defects of these codes are upper bounded by $2^{\frac{m+1}{2}}$.\\

From the second conclusion of Corollary \ref{C-3-1}, we have binary linear four-weight $[2^{2m+1}-2^{m},2m+1,2^{2m}-2^{m-1}-2^{\frac{m-1}{2}}]_2$ code, and binary linear four-weight $[2^{2m+2}-2^{m},2m+2,2^{2m+1}-2^{m-1}-2^{\frac{m-1}{2}}]_2$ code, when $m$ is odd. When $m=3$, we construct explicit four-weight $[120,7,58]_2$ code and four-weight $[248, 8,122]_2$ code. The optimal codes in \cite{Grassl} are linear $[120,7,60]_2$ code and $[248,8,124]_2$ code. Our four-weight binary linear codes are near optimal.\\

Then we have the following conclusion.\\

\begin{theorem}\label{T-5-0}
Let $m\geq 3$ be an odd integer, $t$ be a positive integer.

1)~We construct a family of three-weight minimal $[2^{2m}-2^{m}, 2m, 2^{2m-1}-2^{m-1}-2^{\frac{m-1}{2}}]_2$ codes with the weight distribution in Table \ref{tab-8-0}.

2)~We construct a family of four-weight minimal $[2^{2m+t}-2^{m}, 2m+t, 2^{2m+t-1}-2^{m-1}-2^{\frac{m-1}{2}}]_2$ codes with the weight distribution in Table \ref{tab-8-1}.
\end{theorem}

\begin{proof}{\rm
The minimality follows from the Ashikhmin-Barg criterion, see \cite{AB}.}
\end{proof}

\begin{longtable}{|l|l|}
\caption{\label{tab-8-0} Weight distribution of the codes in Theorem \ref{T-5-0} 1)}     \\ \hline
Weight                               & Weight distribution                  \\ \hline
$0$                                  & $1$                                  \\ \hline
$2^{2m-1}-2^{m-1}-2^{\frac{m-1}{2}}$ & $(2^m-1)(2^{m-2}-2^{\frac{m-3}{2}})$ \\ \hline
$2^{2m-1}-2^{m-1}$                   & $(2^m-1)(2^{m-1}+1)$                 \\ \hline
$2^{2m-1}-2^{m-1}+2^{\frac{m-1}{2}}$ & $(2^m-1)(2^{m-2}+2^{\frac{m-3}{2}})$ \\ \hline
\end{longtable}

\begin{longtable}{|l|l|}
\caption{\label{tab-8-1} Weight distribution of the codes in Theorem \ref{T-5-0} 2)}          \\ \hline
Weight                                 & Weight distribution                     \\ \hline
$0$                                    & $1$                                     \\ \hline
$2^{2m+t-1}-2^{m-1}-2^{\frac{m-1}{2}}$ & $2^t(2^m-1)(2^{m-2}-2^{\frac{m-3}{2}})$ \\ \hline
$2^{2m+t-1}-2^{m-1}$                   & $2^t(2^m-1)(2^{m-1}+1)$                 \\ \hline
$2^{2m+t-1}-2^{m-1}+2^{\frac{m-1}{2}}$ & $2^t(2^m-1)(2^{m-2}+2^{\frac{m-3}{2}})$ \\ \hline
$2^{2m+t-1}$                           & $2^t-1$                                 \\ \hline
\end{longtable}

Kasami codes are a family of $[2^{2m}-1, 3m, 2^{2m-1}-2^{m-1}]_2$ codes constructed in \cite{Kasami}. They are famous for their application in constructing sequences with optimal correlation magnitudes. They have three nonzero weights, $2^{2m-1}-2^{m-1}, 2^{2m-1}, 2^{2m-1}+2^{m-1}$. Let the dual distance of a linear code of the length $n$ be $d'\geq 1$, then the $r$-th generalized Hamming weight $d_r$ of this code is $n-d'$, where $r=k-d'+1$, see \cite[Page 134]{Helleseth}. Then from \cite[Theorem 1]{Helleseth}, the dual distances of Kasami codes, when $m\geq 2$, are at least three. So they are projective linear anticodes. From Corollary 3.1, we have a family of $[2^{3m}-2^{2m}, 3m, 2^{3m-1}-2^{2m-1}-2^{m-1}]_2$ codes with three nonzero weights. Three weights are $2^{3m-1}-2^{2m-1}-2^{m-1}, 2^{3m-1}-2^{2m-1}, 2^{3m-1}-2^{2m-1}+2^{m-1}$.\\

In particular, when $m=2$, we construct an explicit binary linear $[48, 6, 22]_2$ code with three weights $22,24,26$. The optimal linear $[48,6,24]_2$ code was documented in \cite{Grassl}. This binary three-weight code is near optimal.  Griesmer defects of these codes are upper bounded by $2^{m}$.
From Corollary \ref{C-3-1}, we have binary linear four-weight $[2^{3m+1}-2^{2m},3m+1,2^{3m}-2^{2m-1}-2^{m-1}]_2$ codes, and binary linear four-weight $[2^{3m+2}-2^{2m},3m+2,2^{3m+1}-2^{2m-1}-2^{m-1}]_2$ codes. When $m=2$, we construct explicit four-weight $[112,7,54]_2$ and $[240, 8,118]_2$ codes. The optimal ones in \cite{Grassl} are $[112,7,56]_2$ and $[240,8,120]_2$ codes. \\

\begin{theorem}\label{T-5-1}
Let $m\geq 2$ be an integer, $t$ be a positive integer.

1)~ We construct a family of three-weight minimal $[2^{3m}-2^{2m}, 3m, 2^{3m-1}-2^{2m-1}-2^{m-1}]_2$ codes with the weight distribution in Table \ref{tab-8}.

2)~ We construct a family of four-weight minimal $[2^{3m+t}-2^{2m},3m+t,2^{3m+t-1}-2^{2m-1}-2^{m-1}]_2$ codes with the weight distribution in Table \ref{tab-8-3}.
\end{theorem}

\begin{proof}{\rm
The minimality follows from the Ashikhmin-Barg criterion.}
\end{proof}

\begin{longtable}{|l|l|}
\caption{\label{tab-8} Weight distribution of the codes in Theorem \ref{T-5-1} 1)}\\ \hline
Weight                      & Weight distribution   \\ \hline
$0$                         & $1$                   \\ \hline
$2^{3m-1}-2^{2m-1}-2^{m-1}$ & $(2^m-1)2^{2m-1}$     \\ \hline
$2^{3m-1}-2^{2m-1}$         & $2^{2m}-1$            \\ \hline
$2^{3m-1}-2^{2m-1}+2^{m-1}$ & $(2^m-1)2^{2m-1}$     \\ \hline
\end{longtable}

\begin{longtable}{|l|l|}
\caption{\label{tab-8-3} Weight distribution of the codes in Theorem \ref{T-5-1} 2)}\\ \hline
Weight                        & Weight distribution   \\ \hline
$0$                           & $1$                   \\ \hline
$2^{3m+t-1}-2^{2m-1}-2^{m-1}$ & $(2^m-1)2^{2m+t-1}$  \\ \hline
$2^{3m+t-1}-2^{2m-1}$         & $(2^{2m}-1)2^t$    \\ \hline
$2^{3m+t-1}-2^{2m-1}+2^{m-1}$ & $(2^m-1)2^{2m+t-1}$  \\ \hline
$2^{3m+t-1}$                  & $2^t-1$               \\ \hline
\end{longtable}

Let $\frac{m}{l}$ be odd, a family of three-weight binary projective linear $[2^{m-1}-1,m, 2^{m-2}-2^{\frac{m+l-4}{2}}]_2$ codes was constructed in \cite[Theorem 1]{DD1}. The dual distances of these codes were given in \cite[Theorem 7]{DD1}. Three weights are $2^{m-2}-2^{\frac{m+l-4}{2}}, 2^{m-2}, 2^{m-2}+2^{\frac{m+l-4}{2}}$. Then a family of three-weight binary $[2^m-2^{m-1}, m, 2^{m-2}-2^{\frac{m+l-4}{2}}]_2$ codes is obtained from our complementary theorem.
Additionally, the minimality follows from the Ashikhmin-Barg criterion.\\

Three-weight $[16,5,6]_2$, $[32,6,12]_2$ and $[64,7,28]_2$ codes are obtained. The corresponding optimal distance are $8$, $16$ and $32$, see \cite{Grassl}. Then these three-weight binary linear codes are near optimal.
From the second conclusion of our complementary theorem, four-weight linear $[48, 6, 22]_2$, $[96, 7, 44]_2$, $[112, 7, 54]_2$, $[192,8,92]$, $[224, 8, 108]_2$ and  $[240, 8, 118]_2$ codes are constructed. The optimal linear codes in \cite{Grassl} are $[48, 6,24]_2$, $[96, 7, 48]_2$, $[112,7,56]_2$, $[192, 8, 96]_2$, $[224, 8, 112]_2$ and $[240,8,120]_2$ codes.\\

\begin{theorem}\label{T-5-2}
Let $t$, $m$ and $l$ be positive integers, and satisfying $\frac{m}{l}$ is odd.

1)~We construct a family of three-weight minimal $[2^{m-1}, m, 2^{m-2}-2^{\frac{m+l-4}{2}}]_2$ codes with the weight distribution in Table \ref{tab-9}.

2)~We construct a family of four-weight minimal $[2^{m+t}-2^{m-1}, m+t, 2^{m+t-1}-2^{m-2}-2^{\frac{m+l-4}{2}}]_2$ codes with the weight distribution in Table \ref{tab-9-1}.
\end{theorem}

\begin{longtable}{|l|l|}
\caption{\label{tab-9} Weight distribution of the code in Theorem \ref{T-5-2} 1)}\\ \hline
Weight                        & Weight distribution             \\ \hline
$0$                           & $1$                             \\ \hline
$2^{m-2}+2^{\frac{m+l-4}{2}}$ & $2^{m-l-1}+2^{\frac{m-l-2}{2}}$ \\ \hline
$2^{m-2}$                     & $2^m-1-2^{m-l}$                 \\ \hline
$2^{m-2}-2^{\frac{m+l-4}{2}}$ & $2^{m-l-1}-2^{\frac{m-l-2}{2}}$ \\ \hline
\end{longtable}

\begin{longtable}{|l|l|}
\caption{\label{tab-9-1} Weight distribution of the code in Theorem \ref{T-5-2} 2)}\\ \hline
Weight                                  & Weight distribution                  \\ \hline
$0$                                     & $1$                                  \\ \hline
$2^{m+t-1}-2^{m-2}+2^{\frac{m+l-4}{2}}$ & $2^t(2^{m-l-1}+2^{\frac{m-l-2}{2}})$ \\ \hline
$2^{m+t-1}-2^{m-2}$                     & $2^t(2^m-1-2^{m-l})$                 \\ \hline
$2^{m+t-1}-2^{m-2}-2^{\frac{m+l-4}{2}}$ & $2^t(2^{m-l-1}-2^{\frac{m-l-2}{2}})$ \\ \hline
$2^{m+t-1}$                             & $2^t-1$                              \\ \hline
\end{longtable}

Similarly, let $\frac{m}{l}$ be even, a family of three-weight binary projective linear $[2^{m-1}-(-1)^{m\over 2l}2^{{m\over 2}+l-1}-1,m]$ code was also constructed in \cite[Theorem 2]{DD1}. Three weights are $2^{m-2}, 2^{m-2}-(-1)^{m\over 2l}2^{{m\over 2}+l-2}, 2^{m-2}-(-1)^{m\over 2l}2^{{m\over 2}+l-1}$. Then a family of three-weight binary $[2^m-2^{m-1}+(-1)^{m\over 2l}2^{{m\over 2}+l-1}, m]_2$ code is obtained from Corollary \ref{C-3-1}. \\

For example, when $m=8$, $l=1$, three-weight linear $[144,8,64]_2$ code is obtained. The corresponding optimal distance is $70$, see \cite{Grassl}.
From the second conclusion of our complementary theorem, when $m=6$, $l=1$, four-weight linear $[88,7,40]_2$ and $[216,8,104]_2$ codes are constructed. The optimal linear codes in \cite{Grassl} are $[88,7,42]_2$ and $[216,8,107]_2$ codes.\\

\begin{theorem}\label{T-5-3}
    Let $t$, $m$ and $l$ be positive integers, and satisfying $\frac{m}{l}$ is even.

    1)~ If $\frac{m}{l}\equiv 0 \pmod 4$, we construct a family of three-weight minimal $[2^{m-1}+2^{{m\over 2}+l-1}, m, 2^{m-2}]_2$ codes. The weight distribution is given in Table \ref{tab-10}.

    2) We construct a family of four-weight minimal $[2^{m+t}- 2^{m-1}+(-1)^{\frac{m}{2l}}2^{{m\over 2}+l-1}, m+t]_2$ codes with the weight distribution in Table \ref{tab-10-1}.
\end{theorem}

\begin{longtable}{|l|l|}
\caption{\label{tab-10} Weight distribution of the code in Theorem \ref{T-5-3} 1)}  \\ \hline
Weight                       & Weight distribution                               \\ \hline
$0$                          & $1$                                               \\ \hline
$2^{m-2}$                    & $2^{m-2l-1}-1-2^{{m\over 2}-l-1}$                 \\ \hline
$2^{m-2}+2^{{m\over 2}+l-2}$ & $2^m-2^{m-2l}$                                    \\ \hline
$2^{m-2}+2^{{m\over 2}+l-1}$ & $2^{m-2l-1}+2^{{m\over 2}-l-1}$                   \\ \hline
\end{longtable}

\begin{longtable}{|l|l|}
\caption{\label{tab-10-1} Weight distribution of the code in Theorem \ref{T-5-3} 2)}  \\ \hline
Weight                                                 & Weight distribution                                    \\ \hline
$0$                                                    & $1$                                                    \\ \hline
$2^{m+t-1}-2^{m-2}$                                    & $2^t(2^{m-2l-1}-1-(-1)^{m\over 2l}2^{{m\over 2}-l-1})$ \\ \hline
$2^{m+t-1}-2^{m-2}+(-1)^{m\over 2l}2^{{m\over 2}+l-2}$ & $2^t(2^m-2^{m-2l})$                                    \\ \hline
$2^{m+t-1}-2^{m-2}+(-1)^{m\over 2l}2^{{m\over 2}+l-1}$ & $2^t(2^{m-2l-1}+(-1)^{m\over 2l}2^{{m\over 2}-l-1})$   \\ \hline
$2^{m+t-1}$                                            & $2^t-1$                                                \\ \hline
\end{longtable}

Let $m$ be an even positive integer, two families of binary projective linear code $\mathbf{C}'_1$ with parameters $[2^{m-1}-2^{\frac{m-2}{2}}, m+1, 2^{m-2}-2^{\frac{m-2}{2}}]_2$ and code $\mathbf{C}'_2$ with parameters $[2^{m-1}+2^{\frac{m-2}{2}}, m+1, 2^{m-2}]_2$ were constructed, see \cite[Theorem 14.13]{DingTang}. The weights of $\mathbf{C}'_1$ are $2^{m-2}-2^{\frac{m-2}{2}}, 2^{m-2}, 2^{m-1}-2^{\frac{m-2}{2}}$, and three weights of $\mathbf{C}'_2$ are $2^{m-2}, 2^{m-2}+2^{\frac{m-2}{2}}, 2^{m-1}+2^{\frac{m-2}{2}}$. Then two families of three-weight binary $[2^{m+1}-1-2^{m-1}+2^{\frac{m-2}{2}}, m+1, 2^{m}-2^{m-1}+2^{\frac{m-2}{2}}]_2$ and $[2^{m+1}-1-2^{m-1}-2^{\frac{m-2}{2}}, m+1, 2^{m}-2^{m-1}-2^{\frac{m-2}{2}}]_2$ codes are obtained from Corollary \ref{C-3-1}. \\

For example, when $m=4$, the three-weight linear $[21,5,6]_2$ and $[25,5,10]_2$ codes are obtained. The corresponding optimal minimum distances are $10$ and $12$, see \cite{Grassl}. Similarly, four-weight linear $[53,6,22]_2$, $[57,6,26]_2$, $[117,7,54]_2$, $[121,7,58]_2$, $[245,8,118]_2$ and $[249,8,122]_2$ are obtained. The optimal linear codes in \cite{Grassl} have parameters $[53,6,26]_2$, $[57,6,28]_2$, $[117,7,58]_2$, $[121,7,60]_2$, $[245,8,122]_2$ and $[249,8,124]_2$, respectively. Then these binary linear codes are near optimal.\\

Then we have the following conclusion.\\

\begin{theorem}\label{T-5-4}
    Let $m$ be even, $t$ be a positive integer.

    1)~Two families of three-weight minimal codes $\mathbf{C}_1$ with parameters $[2^{m+1}-2^{m-1}+2^{\frac{m-2}{2}}-1, m+1, 2^{m-1}+2^{\frac{m-2}{2}}]_2$ and $\mathbf{C}_2$ with parameters $[2^{m+1}-1-2^{m-1}-2^{\frac{m-2}{2}}, m+1, 2^{m-1}-2^{\frac{m-2}{2}}]_2$ are constructed. The weight distributions are given in Table \ref{tab-11}.

    2)~Two families of four-weight minimal codes $\mathbf{C}_3$ with parameters $[2^{m+t+1}- 2^{m-1}- 2^{\frac{m-2}{2}}-1, m+t+1, 2^{m+t}-2^{m-2}-2^{\frac{m-2}{2}}]_2$ and $\mathbf{C}_4$ with parameters $[2^{m+t+1}- 2^{m-1}+ 2^{\frac{m-2}{2}}-1, m+t+1, 2^{m+t}-2^{m-2}-2^{\frac{m-2}{2}}]_2$ are constructed. The weight distributions are shown in Table \ref{tab-11-1}.
\end{theorem}

\begin{longtable}{|l|l|l|}
\caption{\label{tab-11} Weight distribution of the code in Theorem \ref{T-5-4} 1)}\\ \hline
The weights of $\mathbf{C}_1$   & The weights of $\mathbf{C}_2$   & Weight distribution                          \\ \hline
$0$                             & $0$                             & $1$                                          \\ \hline
$2^m-2^{m-2}+2^{\frac{m-2}{2}}$ & $2^m-2^{m-2}$                   & $2^m-1$                                      \\ \hline
$2^m-2^{m-2}$                   & $2^m-2^{m-2}-2^{\frac{m-2}{2}}$ & $2^m-1$                                      \\ \hline
$2^{m-1}+2^{\frac{m-2}{2}}$     & $2^{m-1}-2^{\frac{m-2}{2}}$     & $1$                                          \\ \hline
\end{longtable}

\begin{longtable}{|l|l|l|}
\caption{\label{tab-11-1} Weight distribution of the codes in Theorem \ref{T-5-4} 2)}\\ \hline
The weights of $\mathbf{C}_3$       &The weights of $\mathbf{C}_4$          & Weight distribution   \\ \hline
$0$                                 &$0$                                    & $1$                   \\ \hline
$2^{m+t}-2^{m-2}+2^{\frac{m-2}{2}}$ &$2^{m+t}-2^{m-2}$                      & $2^t(2^m-1)$          \\ \hline
$2^{m+t}-2^{m-2}$                   &$2^{m+t}-2^{m-2}-2^{\frac{m-2}{2}}$    & $2^t(2^m-1)$          \\ \hline
$2^{m+t}-2^{m-1}+2^{\frac{m-2}{2}}$ &$2^{m+t}-2^{m-1}-2^{\frac{m-2}{2}}$    & $2^t$                 \\ \hline
$2^{m+t}$                           &$2^{m+t}$                              & $2^t-1$               \\ \hline
\end{longtable}


Let $m$ be an odd positive integer, three families of projective linear $[n', m, \frac{n'-2^{(m-1)/2}}{2}]_2$ code with three nonzero weights, $\frac{n'-2^{(m-1)/2}}{2}$, $\frac{n'}{2}$, $\frac{n'+2^{(m-1)/2}}{2}$, were constructed in \cite[Corollary 11]{Ding}, where $n'=2^{m-1}-2^{(m-1)/2}$, $2^{m-1}$ or $2^{m-1}+2^{(m-1)/2}$. For $n'=2^{m-1}-2^{(m-1)/2}$, from Corollary \ref{C-3-1}, we have a family of linear $[2^{m-1}+2^{(m-1)/ 2}-1, m, 2^{m-2}]_2$ codes. For example, when $m=5,7$, there are binary linear three-weight codes $[19, 5, 8]_2$ and $[71,7,32]_2$. The liner code $[19, 5, 8]_2$ is optimal, and the optimal distance of $[71,7]_2$ code is 34, see \cite{Grassl}.\\

From the second conclusion of our complementary theorem, four-weight linear codes with parameters $[9,4,4]_2$, $[11,4,5]_2$, $[25,5,12]_2$, $[27,5,13]_2$, $[43,6,20]_2$, $[51,6,24]_2$, $[57,6,28]_2$, $[59,6,29]_2$, $[107,7,52]_2$, $[115,7,56]_2$, $[121,7,60]_2$, $[123,7,61]_2$, $[235,8,116]_2$, $[243,8,120]_2$, $[249,8,124]_2$ and $[251,8,125]_2$ are constructed, all of these codes are optimal. Moreover four-weight linear $[47,6,22]_2$, $[111,7,54]_2$, $[183,8,88]_2$, $[191,8,92]_2$, $[199,8,96]_2$ and $[239,8,118]_2$ are constructed, and the optimal distances of corresponding codes are 23, 55, 90, 95, 98 and 119, respectively, see \cite{Grassl}.\\

\begin{theorem}\label{T-5-9} Let $m$ be odd. Let $n'=2^{m-1}-2^{(m-1)/2}$, $2^{m-1}$ or $2^{m-1}+2^{(m-1)/2}$.

1)~ The above projective $[2^{m-1}+2^{(m-1)/ 2}-1, m, 2^{m-2}]_2$ code is a minimal three-weight code. The weight distribution is as in Table \ref{tab-17}.

2)~ Let $t\geq1$ be an integer. The four-weight minimal code with parameters $[2^{m+t}-n'-1, m+t, 2^{m+t-1}-\frac{n'+2^{(m-1)/2}}{2}]_2$ is constructed. The weight distribution is as in Table \ref{tab-17-1}.\\
\end{theorem}

\begin{longtable}{|l|l|}
\caption{\label{tab-17} Weight distribution of the codes in Theorem \ref{T-5-9} 1)}\\ \hline
Weight                 & Weight distribution     \\ \hline
$0$                    & $1$                     \\ \hline
$2^{m-2}+2^{(m-1)/2}$  & $2^{m-2}-2^{(m-3)/2}$   \\ \hline
$2^{m-2}+2^{(m-3)/2}$  & $2^{m-1}$               \\ \hline
$2^{m-2}$              & $2^{m-2}+2^{(m-3)/2}-1$ \\ \hline
\end{longtable}

\begin{longtable}{|l|l|}
\caption{\label{tab-17-1} Weight distribution of the codes in Theorem \ref{T-5-9} 2)}\\ \hline
Weight                                & Weight distribution                          \\ \hline
$0$                                   & $1$                                          \\ \hline
$2^{m+t-1}-\frac{n'-2^{(m-1)/2}}{2}$  & $2^t(n'(2^m-n')2^{-m}-n'2^{-(m+1)/2})$       \\ \hline
$2^{m+t-1}-\frac{n'}{2}$              & $2^t(2^m-1-n'(2^m-n')2^{-(m-1)})$            \\ \hline
$2^{m+t-1}-\frac{n'+2^{(m-1)/2}}{2}$  & $2^t(n'(2^m-n')2^{-m}+n'2^{-(m+1)/2})$       \\ \hline
$2^{m+t-1}$                           & $2^t-1$                                      \\ \hline
\end{longtable}

Let $m\geq5$ be an odd positive integer, then a family of projective linear $[2^{m-2}, m-1, 2^{m-3}-2^{\frac{m-3}{2}}]_2$ code with three nonzero weights, $2^{m-3}-2^{\frac{m-3}{2}}, 2^{m-3}, 2^{m-3}+2^{\frac{m-3}{2}}$, was constructed in \cite[Theorem 3]{Heng}. Then from Corollary \ref{C-3-1}, we have a family of $[2^{m+t-1}-2^{m-2}-1, m+t-1, 2^{m+t-2}-2^{m-3}-2^{\frac{m-3}{2}}]_2$ codes, where $t$ is a positive integer. For example, when $m=5, 7$,  binary linear four-weight $[23,5,10]_2$, $[55,6,26]_2$,  $[119,7,58]_2$, $[247,8,122]_2$, $[95,7,44]_2$ and $[223,8,108]_2$ codes are constructed. The optimal distances of corresponding binary linear codes are 11,27,59,123,47 and 111, respectively, see \cite{Grassl}.\\

\begin{theorem}\label{T-5-6}
Let $m\geq 5$ be an odd integer, $t$ be an positive integer. Then we construct a four-weight minimal $[2^{m+t-1}-2^{m-2}-1, m+t-1, 2^{m+t-2}-2^{m-3}-2^{\frac{m-3}{2}}]_2$ codes. The weight distribution is as in the following table.
\end{theorem}

\begin{longtable}{|l|l|}
\caption{\label{tab-13} Weight distribution of the code in Theorem \ref{T-5-6}} \\ \hline
Weight                                & Weight distribution                \\ \hline
$0$                                   & $1$                                \\ \hline
$2^{m+t-2}-2^{m-3}+2^{\frac{m-3}{2}}$ & $2^t(2^{m-4}-2^{{m-5}\over 2})$    \\ \hline
$2^{m+t-2}-2^{m-3}$                   & $2^t(3\cdot 2^{m-3}-1)$            \\ \hline
$2^{m+t-2}-2^{m-3}-2^{\frac{m-3}{2}}$ & $2^t(2^{m-4}+2^{{m-5}\over 2})$    \\ \hline
$2^{m+t-2}$                           & $2^t-1$                            \\ \hline
\end{longtable}

Let $m\geq3$, two-weight binary projective linear $[2^{2m-3}+2^{m-2}-1,2m-2,2^{2m-4}]_2$ codes were given in \cite[Theorem 5.2, Corollary 5.4]{WZZ}.Two weights are $2^{2m-4}$ and $2^{2m-4}+2^{m-2}$.  Then from Corollary \ref{C-3-1}, let $t$ be a positive integer, we have a family of three-weight projective linear $[2^{2m+t-2}-2^{2m-3}-2^{m-2},2m+t-2,2^{2m+t-3}-2^{2m-4}-2^{m-2}]_2$ code. When $m=3$, we obtain linear three-weight $[22,5,10]_2$, $[54,6,26]_2$, $[118,7,58]_2$ and $[246,8,122]_2$ codes. All of these codes are optimal, see \cite{Grassl}.\\

\begin{theorem}\label{T-5-15}
Let $m\geq3$, $t$ be a positive integer. We construct a family of three-weight minimal $[2^{2m+t-2}-2^{2m-3}-2^{m-2},2m+t-2,2^{2m+t-3}-2^{2m-4}-2^{m-2}]_2$ codes. The weight distribution is shown in the following table.
\end{theorem}

\begin{longtable}{|l|l|}
\caption{ Weight distribution of the code in Theorem \ref{T-5-15}} \\ \hline
Weight                                & Weight distribution                \\ \hline
$0$                                   & $1$                                \\ \hline
$2^{2m+t-3}-2^{2m-4}$                 & $2^t(2^{2m-3}+2^{m-2}-1)$          \\ \hline
$2^{2m+t-3}-2^{2m-4}-2^{m-2}$         & $2^t(2^{2m-3}-2^{m-2})$            \\ \hline
$2^{2m+t-3}$                          & $2^t-1$                            \\ \hline
\end{longtable}

Let $m\geq 3$ be an odd positive integer, a family of ternary three-weight projective linear $[\frac{3^m-1}{2},2m, 3^{m-1}-3^{\frac{m-1}{2}}]_3$ codes was constructed, see \cite[Chapter 8.5]{DingTang}. Three weights are $3^{m-1}-3^{\frac{m-1}{2}}$, $3^{m-1}$, $3^{m-1}+3^{\frac{m-1}{2}}$. Then from our complementary theorem, a family of three-weight ternary $[\frac{3^{2m}-3^m}{2}, 2m, 3^{2m-1}-3^{m-1}-3^{\frac{m-1}{2}}]_3$ codes is constructed. Three weights are $3^{2m-1}-3^{m-1}-3^{\frac{m-1}{2}}, 3^{2m-1}-3^{m-1}$ and $3^{2m-1}-3^{m-1}+3^{\frac{m-1}{2}}$.
Similarly, the minimality follows from the Ashikhmin-Barg criterion.\\

When $m=3$, a three-weight $[351,6,231]_3$ code is obtained. The Griesmer defect of this code is $4$. This three-weight ternary code is near optimal.\\

\begin{theorem}\label{T-5-5}
    Let $m\geq 3$ be odd.
    We construct a family of three-weight ternary minimal $[\frac{3^{2m}-3^m}{2}, 2m, 3^{2m-1}-3^{m-1}-3^{\frac{m-1}{2}}]_3$ codes with the weight distribution as in the following table.
\end{theorem}

\begin{longtable}{|l|l|}
\caption{\label{tab-12} Weight distribution of the code in Theorem \ref{T-5-5}}             \\ \hline
Weight                               & Weight distribution                                  \\ \hline
$0$                                  & $1$                                                  \\ \hline
$3^{2m-1}-3^{m-1}+3^{\frac{m-1}{2}}$ & $\frac{1}{2}\cdot(3^{m-1}+3^{\frac{m-1}{2}})(3^m-1)$ \\ \hline
$3^{2m-1}-3^{m-1}$                   & $(2\cdot3^{m-1}+1)(3^m-1)$                           \\ \hline
$3^{2m-1}-3^{m-1}-3^{\frac{m-1}{2}}$ & $\frac{1}{2}\cdot(3^{m-1}-3^{\frac{m-1}{2}})(3^m-1)$ \\ \hline
\end{longtable}

Let $m$ be an odd positive integer, then a family of projective linear $[\frac{3^{3m-1}-1}{2}, 3m, 3^{3m-2}-3^{2m-2}]_3$ codes with three nonzero weights, $3^{3m-2}-3^{2m-2}$, $3^{3m-2}$, $3^{3m-2}+3^{2m-2}$, was constructed in \cite[Theorem 16]{Ding}. From our complementary theorem, we have a family of ternary $[3^{3m-1}, 3m, 2\cdot3^{3m-2}-3^{2m-2}]_3$ codes. For example, when $m=1$, a ternary linear three-weight $[9,3,5]_3$ code is constructed. The optimal distance of the linear $[9,3]_3$ code is 6, see \cite{Grassl}.\\

\begin{theorem}\label{T-5-10}
Let $m$ be odd. The above projective ternary $[3^{3m-1}, 3m, 2\cdot3^{3m-2}-3^{2m-2}]_3$ code is a minimal three-weight code. The weight distribution is as in the following table.
\end{theorem}

\begin{longtable}{|l|l|}
\caption{\label{tab-18} Weight distribution of the three-weight code in Theorem \ref{T-5-10}}\\ \hline
Weight&Weight distribution \\ \hline
$0$                        & $1$                     \\ \hline
$2\cdot3^{3m-2}+3^{2m-2}$  & $3^{2m}+3^m$            \\ \hline
$2\cdot3^{3m-2}$           & $3^{3m}-2\cdot3^{2m}-1$ \\ \hline
$2\cdot3^{3m-2}-3^{2m-2}$  & $3^{2m}-3^m$            \\ \hline
\end{longtable}

When $q$ is a prime power, then two-weight $q$-ary $[q^2+1, 4, q^2-q]_q$ codes were constructed in \cite[Theorem 13.6]{DingTang}. Two weights are $q^2-q$ and $q^2$. These are well-known ovoid codes. Its dual distance is $4$, see \cite[Theorem 13.5]{DingTang}. From  Corollary \ref{C-3-1}, we have the following result.\\

\begin{theorem}\label{T-5-7}
Let $q$ be a prime power. We construct a family of two-weight minimal $[q^3+q, 4, q^3-q^2]_q$ codes. The weight distribution is shown in Table \ref{tab-14}.
\end{theorem}

\begin{longtable}{|l|l|}
\caption{\label{tab-14} Weight distribution of the two-weight $[q^3+q,4,q^3-q^2]_q$ code}\\ \hline
Weight&Weight distribution \\ \hline
$0$&$1$ \\ \hline
$q^3-q^2$&$(q-1)(q^2+1)$ \\ \hline
$q^3-q^2+q$&$(q^2-q)(q^2+1)$ \\ \hline
\end{longtable}

For example, when $q=4$, we get a two-weight $[68,4,48]_4$ code. The optimal code in \cite{Grassl} is a linear $[68,4,50]_4$ code. When $q=5$, we get a two-weight $[130,4,100]_5$ code. The optimal code in \cite{Grassl} is a linear $[130,4,103]_5$ code.\\

Let $S_1 \subset {\bf F}_q^4$ and $S_2 \subset {\bf F}_q^4$ be two linear subspace of the dimension two satisfying $S_1 \bigcap S_2={\bf 0}$. Then using $\frac{q^2-1}{q-1}+\frac{q^2-1}{q-1}$ columns in these two subspaces as one generator matrix, we have a projective $[2(q+1),4]_q$ code ${\bf C}$. Any hyperplane in ${\bf F}_q^4$ intersects such $2(q+1)$ columns at $2$ elements, or $q+1+1$ elements. Then the code ${\bf C}$ is a two-weight $[2q+2,4,q]_q$ code. Two weights are $q$ and $2q$.  The complementary code is a projective linear $[q^3+q^2-q-1,4, q^3-2q]_q$ code. This is a two-weight linear code. Two weights are $q^3-2q$ and $q^3-q$. When $q=3$, this is a two-weight $[32,4,21]_3$ code. When $q=4$, this is a two-weight $[75,4,56]_4$ code. Both codes have optimal parameters, see \cite{Grassl}.\\

In the following table, we list some near-optimal two-weight binary linear codes from the above construction.

\begin{longtable}{|l|l|l|l|}
\caption{\label{tab-1-2} Two-weight codes}\\ \hline
$q$&Parameters&Weights&Optimal parameters\\ \hline
$5$ &$[12,4,5]_5$&$5,10$ &$[12,4,8]_5$ \\ \hline
$7$ &$[16,4,7]_7$&$7,14$&$[16,4,11]_7$ \\ \hline
$8$ &$[18,4,8]_8$&$8,16$&$[18,4,13]_8$ \\ \hline
$9$ &$[20,4,9]_9$&$9,18$&$[20,4,15]_9$ \\ \hline
\end{longtable}

\begin{theorem}\label{T-5-8}
1) The above projective linear $[2(q+2),4,q]_2$ code has the following weight distribution.

\begin{longtable}{|l|l|}
\caption{\label{tab-15} Weight distribution of the two-weight $[2q+2,4,q]_q$ code}\\ \hline
Weight&Weight distribution \\ \hline
$0$  & $1$         \\ \hline
$q$  & $2(q^2 -1)$ \\ \hline
$2q$ & $(q^2-1)^2$ \\ \hline
\end{longtable}

2) The above projective linear $[q^3+q^2-q-1,4,q^3-2q]_q$ code is a minimal two-weight code. The weight distribution is as in Table \ref{tab-16}.
\end{theorem}

\begin{longtable}{|l|l|}
\caption{\label{tab-16} Weight distribution of the two-weight $[q^3+q^2-q-1,4,q^3-2q]_q$ code}\\ \hline
Weight&Weight distribution \\ \hline
$0$      & $1$          \\ \hline
$q^3-q$  & $2(q^2-1)$   \\ \hline
$q^3-2q$ & $(q^2-1)^2$  \\ \hline
\end{longtable}

Let $m \geq 2$, $q$ be an prime, then a family of projective linear $[q^{2m}+1,3m,q^{m-1}(q^{m+1}-q^{m}-1)]_q$ codes with eight nonzero weights $q^{m-1}(q^{m+1}-q^{m}-1)$, $q^{m-1}(q^{m+1}-q^{m}-1)+1$, $q^{m-1}(q^{m+1}-q^{m}-1)+2$, $q^{2m-1}(q-1)$, $q^{2m-1}(q-1)+1$, $(q-1)(q^{2m-1}+q^{m-1})$, $(q-1)(q^{2m-1}+q^{m-1})+1$, $(q-1)(q^{2m-1}+q^{m-1})+2$ was constructed in \cite[Theorem VI.4]{HengDing}. From Corollary \ref{C-3-1}, we have a family of $q$-ary $[\frac{q^{3m}-1}{q-1}-q^{2m}-1, 3m, q^{3m-1}-(q-1)(q^{2m-1}+q^{m-1})-2]_q$ codes. For example, when $q=2$, $m=2$,  linear eight-weight $[46,6,20]_2$ code is constructed.  The optimal distance of $[46, 6]_2$ code is 22, see \cite{Grassl}. When $q=3$, $m=2$, a linear $[282, 6, 181]_3$ code is constructed. The Griesmer defect of this code is $8$.\\

\begin{theorem}\label{T-5-11}
Let $m \geq 2$ be an integer, $q$ be an prime. The projective $q$-ary $[\frac{q^{3m}-1}{q-1}-q^{2m}-1, 3m, q^{3m-1}-(q-1)(q^{2m-1}+q^{m-1})-2]_q$ code is a minimal eight-weight code. The weight distribution is as in Table \ref{tab-19}.
\end{theorem}

\begin{longtable}{|l|l|}
\caption{\label{tab-19} Weight distribution of the eight-weight code in Theorem \ref{T-5-11}}\\ \hline
Weight&Weight distribution \\ \hline
$0$                                 & $1$                                   \\ \hline
$q^{3m-1}-q^{m-1}(q^{m+1}-q^m-1)$   & $(q-1)(q^{m-1}-1)(q^{2m-2}+q^{m-1})$  \\ \hline
$q^{3m-1}-q^{m-1}(q^{m+1}-q^m-1)-1$ & $q^{2m-2}(q-1)^{2}(2q^{m-1}-1)$       \\ \hline
$q^{3m-1}-q^{m-1}(q^{m+1}-q^m-1)-2$ & $q^{2m-2}(q-1)^{2}(q^{m}-q^{m-1}+1)$  \\ \hline
$q^{3m-1}-q^{2m-1}(q-1)$            & $q^{2m-1}-1$                          \\ \hline
$q^{3m-1}-q^{2m-1}(q-1)-1$          & $q^{2m}-q^{2m-1}$                     \\ \hline
$q^{3m-1}-(q-1)(q^{2m-1}+q^{m-1})$  & $(q^{m-1}-1)(q^{2m-2}-q^{m}+q^{m-1})$ \\ \hline
$q^{3m-1}-(q-1)(q^{2m-1}+q^{m-1})-1$& $q^{2m-2}(q-1)(2q^{m-1}-1)$           \\ \hline
$q^{3m-1}-(q-1)(q^{2m-1}+q^{m-1})-2$& $q^{2m-2}(q-1)^{2}(q^{m-1}-1)$        \\ \hline
\end{longtable}

Let $m$ be an odd positive integer, $q$ be an odd prime.
A family of projective linear $[{q^{m}+1,2m,q^{m-1}(q-1)-q^{\frac {m-1}{2}}}]_q$ code with nine nonzero weights,
$q^{m-1}(q-1)$, $q^{m-1}(q-1)+1$, $q^{m-1}(q-1)+2$, $q^{m-1}(q-1)-q^{\frac {m-1}{2}}(-1)^{\frac {(q-1)(m+1)}{4}}$, $q^{m-1}(q-1)-q^{\frac {m-1}{2}}(-1)^{\frac {(q-1)(m+1)}{4}}+1$,
$q^{m-1}(q-1)-q^{\frac {m-1}{2}}(-1)^{\frac {(q-1)(m+1)}{4}}+2$, $q^{m-1}(q-1)+q^{\frac {m-1}{2}}(-1)^{\frac {(q-1)(m+1)}{4}}$,
$q^{m-1}(q-1)+q^{\frac {m-1}{2}}(-1)^{\frac {(q-1)(m+1)}{4}}+1$,
$q^{m-1}(q-1)+q^{\frac {m-1}{2}}(-1)^{\frac {(q-1)(m+1)}{4}}+2$, was constructed in \cite[Theorem VI.7]{HengDing}. From our complementary theorem, we have a family of $q$-ary linear $[\frac{q^{2m}-1}{q-1}-q^m-1, 2m, q^{2m-1}-q^{m-1}(q-1)-q^{(m-1)/2}-2]_q$ code.
For example, when $q=3$ and $m=3$, a ternary linear $[336, 6, 220]_3$ codes with weights from 220 to 228 is constructed. In fact, this code is near the optimal code since the Griesmer defect of this code is 4.\\

\begin{theorem}\label{T-5-12}
Let $m$ be an odd positive integer, $q$ be an odd prime. The projective $q$-ary $[\frac{q^{2m}-1}{q-1}-q^m-1, 2m, q^{2m-1}-q^{m-1}(q-1)-q^{(m-1)/2}-2]_q$ code is a nine-weight minimal code. The weight distribution is as in Table \ref{tab-20}.
\end{theorem}

\begin{longtable}{|l|l|}
\caption{\label{tab-20} Weight distribution of the nine-weight code in Theorem \ref{T-5-12}}\\ \hline
Weight&Weight distribution \\ \hline
$0$                        & $1$                     \\ \hline
$q^{2m-1}-q^{m-1}(q-1)$& $(q^{m-1}-1)(q^{m-2}+1)$            \\ \hline
$q^{2m-1}-q^{m-1}(q-1)-1$& $q^{m-2}(q-1)(2q^{m-1}+2q-2)$            \\ \hline
$q^{2m-1}-q^{m-1}(q-1)-2$& $q^{m-2}(q-1)^{2}(q^{m-1}-1)$            \\ \hline
$q^{2m-1}-q^{m-1}(q-1)+q^{\frac {m-1}{2}}(-1)^{\frac {(q-1)(m+1)}{4}}$& $\frac{(q-1)(q^{m-1}-1)({q^{m-2}+(-1)^{\frac {(q-1)(m+1)}{4}}q^{\frac {m-1}{2}}})}{2}$            \\ \hline
$q^{2m-1}-q^{m-1}(q-1)+q^{\frac {m-1}{2}}(-1)^{\frac {(q-1)(m+1)}{4}}-1$& $q^{m-2}(q-1)^{2}(q^{m-1}-1)$            \\ \hline
$q^{2m-1}-q^{m-1}(q-1)+q^{\frac {m-1}{2}}(-1)^{\frac {(q-1)(m+1)}{4}}-2$& $\frac{(q-1)^{2}[{q^{m-2}(q^{m}-q^{m-1}+1)+(-1)^{\frac {(q-1)(m+1)}{4}}q^{\frac {3(m-1)}{2}}}]}{2}$            \\ \hline
$q^{2m-1}-q^{m-1}(q-1)-q^{\frac {m-1}{2}}(-1)^{\frac {(q-1)(m+1)}{4}}$& $\frac{(q-1)(q^{m-1}-1)({q^{m-2}+(-1)^{\frac {(q-1)(m+1)+4}{4}}q^{\frac {m-1}{2}}})}{2}$            \\ \hline
$q^{2m-1}-q^{m-1}(q-1)-q^{\frac {m-1}{2}}(-1)^{\frac {(q-1)(m+1)}{4}}-1$& $q^{m-2}(q-1)^{2}(q^{m-1}-1)$            \\ \hline
$q^{2m-1}-q^{m-1}(q-1)-q^{\frac {m-1}{2}}(-1)^{\frac {(q-1)(m+1)}{4}}-2$& $\frac{(q-1)^{2}[{q^{m-2}(q^{m}-q^{m-1}+1)+(-1)^{\frac {(q-1)(m+1)+4}{4}}q^{\frac {3(m-1)}{2}}}]}{2}$            \\ \hline
\end{longtable}

Let $m\geq 3$ be an odd integer, $q$ be an odd prime.
Then there is a projective $q$-ary linear $[\frac{q^{m-1}-1}{q-1}, m, q^{m-2}-q^{{m-3}\over 2}]_q$ code with three weights $q^{m-2}-q^{{m-3}\over 2}$, $q^{m-2}$ and $q^{m-2}+q^{{m-3}\over 2}$ was constructed in \cite[Corollary 3]{DD}. Then there exists a linear $[q^{m-1}, m, q^{m-1}-q^{m-2}-q^{{m-3}\over 2}]_q$ code with three weights.
For example, when $q=3$, $m=3, 5$, ternary $[9,3,5]_3$ and $[81, 5, 51]_3$ codes are constructed. When $q=5$, $m=3, 5$, the quinary $[25,3,19]_5$ and $[625,5,495]_5$ codes are constructed. The optimal distances of $[9,3]_3$, $[81, 5]_3$, and $[25,3]_5$ codes are 6, 54, and 20 respectively.  The linear $[625,5,495]_5$ code has the Griesmer defect $6$.\\

\begin{theorem}\label{T-5-13}
Let $m\geq 3$ be an odd integer, $q$ be an odd prime. The projective $q$-ary $[q^{m-1}, m, q^{m-1}-q^{m-2}-q^{{m-3}\over 2}]_q$ code is a minimal three-weight code. The weight distribution is shown in Table \ref{tab-21}.
\end{theorem}

\begin{longtable}{|l|l|}
\caption{\label{tab-21} Weight distribution of the three-weight code in Theorem \ref{T-5-13}}\\ \hline
Weight&Weight distribution \\ \hline
$0$                                   & $1$              \\ \hline
$q^{m-1}-q^{m-2}+q^{{m-3}\over 2}$    & $\frac{q-1}{2}(q^{m-1}+q^{{m-1}\over 2})$     \\ \hline
$q^{m-1}-q^{m-2}$                     & $q^{m-1}-1$ \\ \hline
$q^{m-1}-q^{m-2}-q^{{m-3}\over 2}$    & $\frac{q-1}{2}(q^{m-1}-q^{{m-1}\over 2})$            \\ \hline
\end{longtable}

Let $m\geq 2$ be even, $q$ be an odd prime.
Then there is a projective $q$-ary linear $[\frac{q^{m-1}-1}{q-1}-(-1)^{(\frac{q-1}{2})^2 \frac{m}{2}}q^{\frac{m-2}{2}}, m]_q$ code with two weights $q^{m-2}$ and $q^{m-2}-(-1)^{(\frac{q-1}{2})^2 \frac{m}{2}}q^{\frac{m-2}{2}}$ was constructed in \cite[Corollary 4]{DD}.
Then we can construct a linear $[q^{m-1}+(-1)^{(\frac{q-1}{2})^2 \frac{m}{2}}q^{\frac{m-2}{2}}, m]_q$ code with two weights.
For example, when $q=3$, $m=4, 6$, there are ternary $[30, 4, 18]_3$ and $[234, 6, 153]_3$ codes. The optimal distances of codes $[30, 4]_3$ and $[234, 6]_3$ are 19 and 155, respectively. When $q=5$, $m=4$, there is a quinary $[130, 4, 100]_5$ code. The optimal code in \cite{Grassl} is a $[130, 4, 103]_5$ code.\\

\begin{theorem}\label{T-5-14}
Let $m\geq 4$ be even, $q$ be an odd prime. The projective $[q^{m-1}+(-1)^{(\frac{q-1}{2})^2 \frac{m}{2}}\cdot q^{\frac{m-2}{2}}, m]_q$ code is a minimal two-weight code. The weight distribution is given in Table \ref{tab-22}.
\end{theorem}

\begin{longtable}{|l|l|}
\caption{\label{tab-22} Weight distribution of the two-weight code in Theorem \ref{T-5-14}}\\ \hline
Weight&Weight distribution \\ \hline
$0$                                                                 & $1$              \\ \hline
$(q-1)q^{m-2}$                                                      & $q^{m-1}-(-1)^{(\frac{q-1}{2})^2 \frac{m}{2}}q^{{m-2}\over 2}(q-1)-1$ \\ \hline
$(q-1)q^{m-2}+(-1)^{(\frac{q-1}{2})^2 \frac{m}{2}}q^{{m-2}\over 2}$ & $(q-1)(q^{m-1}+(-1)^{(\frac{q-1}{2})^2 \frac{m}{2}}q^{{m-2}\over 2})$  \\ \hline
\end{longtable}

To the best of knowledge, most families of few-weight linear codes constructed in this paper have not been reported in the  literature, see \cite[Chapter 4, Chapter 14]{DingTang}, \cite{CG,HengYue,DD1,DD,carlet} and references therein.\\

\section{Few-weight linear codes from concatenated and complementary codes}\label{sec-6}

It is obvious that $t$-weight linear codes can be obtained from the concatenation of $t'$-weight codes over larger fields. For example, $q=2^s$, let the outer code be the ovoid code $[q^2+1,4,q^2-q]_q$ with two weights $q^2-q$ and $q^2$, and the inner code be the the simplex $[2^s-1,s,2^{s-1}]_2$ code with one weight $2^{s-1}$, the concatenation code is a linear $[(2^s-1)(4^s+1), 4s, 2^{s-1}(4^s-2^s)]_2$ code ${\bf C}$. A nonzero codeword in ${\bf C}$ has weight $2^{s-1}(q^2-q)=2^{s-1}(4^s-2^s)$ or $2^{3s-1}$, since each nonzero coordinate in ${\bf F}_{s^s}$ is replaced by a weight $2^{s-1}$ codeword in the inner code. This code is projective, following the dual code described in \cite{CLX}. Therefore we have the following result.\\

\begin{theorem}\label{T-6-1}
Let $s$ be a positive integer, we construct a linear two-weight $[(2^s-1)(4^s+1), 4s, 2^{s-1}(4^s-2^s)]_2$ code. The weight distribution is as in the following table.
\end{theorem}

\begin{longtable}{|l|l|}
\caption{\label{tab-23} Weight distribution of the code in Theorem \ref{T-6-1}}\\ \hline
Weight&Weight distribution \\ \hline
$0$                 & $1$                   \\ \hline
$2^{3s-1}$          & $(2^s-1)(4^s+1)$      \\ \hline
$2^{3s-1}-2^{2s-1}$ & $(4^s-2^s)(4^s+1)$    \\ \hline
\end{longtable}

When $s=2$, ${\bf C}$ is a two-weight linear $[51, 8, 24]_2$ code. Two weights are $24$ and $32$. This is an optimal code, see \cite{Grassl}. Notice that the optimal $[51,8,24]_2$ code in \cite{Grassl} is a special cyclic code with two weights $24$ and $32$.  When $s=3$, ${\bf C}$ is a two-weight linear $[455, 12, 224]_2$ code. Its Griesmer defect is $4$.\\

{\bf Corollary 6.1} {\em The complementary code is a linear two-weight $[2^{4s}-1-(2^s-1)(4^s+1), 4s, 2^{4s-1}-2^{3s-1}]_2$ code.}\\

When $s=2$, the complementary code of ${\bf C}$ is a linear two-weight $[204,8,96]_2$ code. Two weights are $96$ and $104$. The optimal $[204,8]_2$ code in \cite{Grassl} has the minimum distance $100$. \\

Similarly, we have a two-weight linear $[(2^s-1)(2^{s+1}+2),4s, 2^{2s-1}]_2$ code from the outer $[2q+2,4,q]_q$ code in Section 5. Two weights are $2^{2s-1}$ and $2^{2s}$. When $s=2$, this is a linear two-weight $[30,8,8]_2$ code. Two weights are $8$ and $16$.  The optimal distance of $[30,8]_2$ code in \cite{Grassl} is $12$. From the dual code described in \cite{CLX}, this code is projective code. The complementary code is a linear $[225,8,112]_2$ code. Two weights are $112$ and $120$. This is an optimal code, see \cite{Grassl}.\\

\section{Minimal binary linear codes with optimal and near optimal parameters}\label{sec-7}

From projective linear anticodes constructed in Subsection 4.2 and Theorem 3.1, we have the following result. Actually, these codes in the case $t=0$, have been constructed in \cite{SZ} as optimal locally reparable codes.\\

{\bf Theorem 7.1} {\em Let $s$ be a positive integer and $t$ be a nonnegative integer. Then an explicit binary linear $[2^{2s-1+t}-s(2s-1)-1, 2s-1+t, 2^{2s-2+t}-s^2]_2$ is constructed.}\\

For example, when $s=3$, we get a projective linear $[15,5]_2$ code with the maximum weight $9$. Then binary linear $[16,5,7]_2$, $[48,6, 23]_2$, $[112,7,55]_2$ and $[240, 8, 119]_2$ codes are constructed. These codes are almots optimal. When $s=4$, we get a projective linear $[28,7]_2$ code with the maximum weight $16$. Then binary linear $[99,7,48]_2$ and $[227, 8, 112]_2$ codes are constructed. These two binary codes are optimal codes, as documented in \cite{Grassl}. Notice that these codes were not reported in \cite{Farrell}. These codes were reported in \cite[Table 1]{SZ} as optimal locally reparable codes.\\

Similarly, when $k=7$ is odd, we construct a binary linear $[21, 6, 6]_2$ code with the maximum weight $12$. Then we have a binary projective linear $[42,6,20]_2$ code, $[106,7,52]_2$ code, and $[234,8,116]_2$ code. Comparing with \cite{Grassl}, these three binary linear codes are optimal. \\

From the binary three-weight projective linear $[70,7,32]_2$ code constructed in Subsection 4.2, we construct a binary $[185,8,88]_2$ code, the corresponding optimal minimum distance is $91$, see \cite{Grassl}.\\

We show that some binary projective linear codes constructed in this section are actually minimal codes. Let ${\bf C}$ be the binary projective linear $[\frac{k(k-1)}{2},k-1,k-1]_2$ code constructed in  \cite{Yan}. The maximum weight of this code is $\frac{k^2}{2}$, when $k$ is even, or $\frac{k^2-1}{4}$ when $k$ is odd. We have the following result.\\

{\bf Corollary 7.1} {\em Let ${\bf C}_1$ be the complementary $[2^{k-1}-1-\frac{k(k-1)}{2},k-1, 2^{k-2}-\delta({\bf C})]_2$ code described as above. Then ${\bf C}_1$ is a minimal code, when $k \geq 3$. Here $\delta({\bf C})=\frac{k^2}{4}$ if $k$ is even, or $\delta({\bf C})=\frac{k^2-1}{4}$ if $k$ is odd.}\\

{\bf Proof.} The code ${\bf C}_1$ satisfies the Ashikhmin-Barg criterion.\\

{\bf Corollary  7.2} {\em Let ${\bf C}_1$ be the complementary $[2^{k-1+t}-1-\frac{k(k-1)}{2},k-1+t, 2^{k-2+t}-\delta({\bf C})]_2$ code described as above. Then ${\bf C}_1$ is a minimal code. Here $\delta({\bf C})=\frac{k^2}{4}$ if $k$ is even, or $\delta({\bf C})=\frac{k^2-1}{4}$ if $k$ is odd.}\\

{\bf Proof.} The conclusion follows from the Ashikhmin-Barg criterion immediately.\\

From Corollary 7.1 and 7.2, many minimal binary linear codes with optimal or almost optimal parameters are constructed, see Table 3 of Section 1. Most codes constructed in previous section satisfy the Ashikhmin-Barg criterion, they are minimal linear codes.\\

\section{Complementary MDS codes}\label{sec-8}

We observe the Reed-Solomon $[q, k, q+1-k]_q$ code. When $k \geq 2$, this is a projective linear code. Then the complementary code is a $q$-ary linear $[\frac{q^k-1}{q-1}-q, k, q^{k-1}-q]_q$ code. We call this code a complementary Reed-Solomon code. Since the maximum weight is $q^{k-1}-q+k-1$, this code satisfies the Ashikhmin-Barg condition, this is a minimal $q$-ary linear code. On the other hand, we have $$q^{k-1}-q+q^{k-2}-1+q^{k-3}+\cdots+1=\frac{q^k-1}{q-1}-q-1.$$ This is an almost Griesmer code. The diameter of the complementary Reed-Solomon code is $q^{k-1}-q+k-1$. Then the sum $$\Sigma_{i=0}^{k-1} \lfloor \frac{\delta({\bf C})}{q^i}\rfloor=\frac{q^k-1}{q-1}-q+k-1.$$

From complementary theorem, we have the following result.\\

{\bf Theorem 8.1} {\em  Let $k$ be a positive integer and $h$ be a non-negative integer. A complementary Reed-Solomon $[\frac{q^{k+h}-1}{q-1}-q,k+h,q^{k+h-1}-q]_q$ code is constructed. It is a minimal k-weight almost Griesmer code when $h=0$ and it is a $(k+1)$-weight almost Griesmer code when $h>0$. As an anticode, it has a difference $k-1$ to the antiGriesmer bound.}\\

In the table 4 of Section 1, we list some minimal optimal complementary Reed-Solomon code.\\

{\bf Corollary 8.1} {\em Let $q$ be a prime power. Then we construct a three-weight almost Griesmer $[q^2+1,3,q^2-q]_q$ code. Three weights are $q^2-q$, $q^2-q+1$ and $q^2-q+2$.}\\

In \cite[Theorem V.I]{HengDing}, three-weight $[p^m+1, m+1, p^{m-1}(p-1)]_p$ codes were constructed. Three weights are $p^{m-1}(p-1)$, $p^{m-1}(p-1)+1$ and $p^m$. When $m=2$, it is clear that our three-weight codes are different to their codes.\\

{\bf Corollary 8.2} {\em A complementary Reed-Solomon $[q^3+q^2+1,4,q^3-q]_q$ code is a minimal $4$-weight almost Griesmer code. Four weights are $q^3-q$, $q^3-q+1$, $q^3-q+2$ and $q^3-q+3$.}\\

Similarly, we can construct complementary code of the trivial MDS $[k,k,1]_q$ code. Then we have the following result. The first conclusion is actually a special case of projective Solomon-Stiffler codes, see \cite{Solomon}.\\

{\bf Theorem 8.2} {\em  A complementary MDS $[\frac{q^{k}-1}{q-1}-k,k,q^{k-1}-k]_q$ code is a minimal $k$-weight Griesmer code, if $k<q$.  As an anticode, it has a difference $k-1$ to the antiGriesmer bound, if $k>q$. A complementary MDS $[\frac{q^k-1}{q-1}-k,k,q^{k-1}-k]_q$ code is a minimal almost Griesmer $k$-weight code, if $q\leq k<2q-1$. }\\

{\bf Corollary 8.3} {\em Let $q$ be a prime power satisfying $q \geq 4$. Then we construct a minimal three-weight Griesmer $[q^2+q-2, 3, q^2-3]_q$ code.}\\

From Corollary 8.3, we construct optimal minimal linear $[18,3,13]_4$ code, $[28,3,22]_5$ code, $[54,3,46]_7$ code, $[70,3,61]_8$ code and $[88,3,78]_9$ code, see \cite{Grassl}.\\

\section{$l$-strongly walk-regular graphs}\label{sec-9}

In this section, based on binary three-weight projective linear codes constructed in Section 5, we present some $l$-SWRGs by applying the construction in \cite{KKSS}.\\

\begin{lemma}[\cite{ML}, Theorem 5.1]\label{L-9-1}
     Let $\mathbf{C}$ be a binary projective linear three-weight $[n, k]_2$ code with weights $0 = w_0 < w_1 < w_2 < w_3$ satisfying
        $$w_1 + w_2 + w_3 =\frac{3n}{2}.$$
     Then the coset graph $\Gamma_{\mathbf{C}^\bot}$ of the dual code of $\mathbf{C}$ is a 3-SWRG. Furthermore, if
        $$w_2 =\frac{n}{2},$$
     then $\Gamma_{\mathbf{C}^\bot}$ is an $l$-SWRG for every odd $l \geq 3$.
     Moreover, the spectrum of $\Gamma_{\mathbf{C}^\bot}$ is given by $\{(n - 2w_i)^{A_{w_i}} |0 \leq i \leq 3\}$.
     Additionally, the parameters $\lambda_l$, $\mu_l$, $\nu_l$ of $\Gamma_{\mathbf{C}^\bot}$ are related to its eigenvalues in the following way:

     1) The eigenvalues $\{(n - 2w_i) | 1 \leq i \leq 3\}$ are the roots of the equation $x^l + (\mu_l - \lambda_l)x + (\mu_l - \nu_l) = 0$;

     2) The code length $n$ satisfies $n^l + (\mu_l - \lambda_l)n + (\mu_l - \nu_l) = \mu_l v$, where $v$ is the number of vertices of $\Gamma_{\mathbf{C}^\bot}$, that is, the number of cosets of the code $\mathbf{C}^\bot$.\\
\end{lemma}

Then combining Theorem \ref{T-5-0}, Theorem \ref{T-5-1}, Theorem \ref{T-5-2} and Theorem \ref{T-5-3}, as well as Lemma \ref{L-9-1}, we have the following result.\\

\begin{theorem}
    Let $\mathbf{C}$ be a binary linear $[n,k]_2$ code. Let $l\geq 3$ be an odd positive integer. The coset graphs $\Gamma_{\mathbf{C}^\bot}$ of the dual codes of the following binary three-weight projective linear codes $\mathbf{C}$ are $l$-SWRGs.
    The spectrum of $\Gamma_{\mathbf{C}^\bot}$ is given by $$\{n^1,(n - 2w_1)^{A_{w_1}} , 0^{A_{w_2}} , -(n - 2w_1)^{A_{w_3}}\}.$$
    Specifically, when $l=3$, the parameters $(\lambda_l, \mu_l, \nu_l)$ of the 3-SWRG with $2^k$ vertices are given by $\lambda_3 = \mu_3+(n - 2w_1)^2$, $\mu_3 = \nu_3 = \frac{4nw_1(n-w_1)}{2^k}$.

    1) A family of linear $[2^{2m}-2^{m}, 2m, 2^{2m-1}-2^{m-1}-2^{\frac{m-1}{2}}]_2$ codes, where $m\geq 3$ is odd;

    2) A family of linear $[2^{3m}-2^{2m}, 3m, 2^{3m-1}-2^{2m-1}-2^{m-1}]_2$ codes, where $m\geq 2$;

    3) A family of linear $[2^{m-1}, m, 2^{m-2}-2^{\frac{m+r-4}{2}}]_2$ codes, where $\frac{m}{r}$ is odd;

    4) A family of linear $[2^{m-1}+2^{{m\over 2}+r-1}, m, 2^{m-2}]_2$ codes, where $\frac{m}{r}$ is even.
\end{theorem}

\begin{proof}{\rm
    Based on these theorems in Section \ref{sec-5} and their corresponding weight distribution tables, it is easy to verify that the above four families of codes are all three-weight binary projective linear codes. Their three weights satisfy above equations in Lemma 9.1.  Then the conclusion is proved.\\
}
\end{proof}

\begin{example}{\rm
Let $\mathbf{C}$ be a binary three-weight projective linear code with parameters $[56,6,26]_2$, which is constructed in Theorem \ref{T-5-0}. Three weights are $w_1=26$, $w_2=28$, $w_3=30$. The weight distribution is as follows, $A_{w_1} = 7$, $A_{w_2} = 35$, $A_{w_3} = 21$. Then the coset grapy $\Gamma_{\mathbf{C}^\bot}$ is an $l$-SWRG for every $l \geq 3$. The spectrum of $\Gamma_{\mathbf{C}^\bot}$ is $\{56^1, 4^7, 0^{35}, -4^{21}\}$. When $l=3$, the parameters of the 3-SWRG $\Gamma_{\mathbf{C}^\bot}$ with 64 vertices are $(2746, 2730, 2730)$.
}
\end{example}

\section{Conclusions}\label{sec-10}

In this paper, we revisited and extended the idea of linear codes from projective linear anticodes in \cite{Farrell}. We gave an antiGriesmer bound on diameters of projective linear anticodes and gave many infinite families of codes attaining or close to the bound.  A complementary theorem was proved for $t$-weight projective linear codes.  Then from one $t$-weight projective linear code, infinitely many $(t+1)$-weight projective linear codes were constructed and their weight distributions were determined.\\

In summary, the following codes were constructed in this paper.\\

1) Many $t$-weight binary linear codes with new parameters were constructed. Many codes in these families are optimal, almost optimal or near optimal. These codes include almost Griesmer complementary Reed-Solomon codes.\\

2) We constructed many optimal or near optimal minimal linear codes. These codes include complementary MDS codes.\\

3) We gave several families of $l$-strongly walk-regular graphs for each odd $l\geq 3$, from complementary three-weight binary projective linear codes constructed in this paper.\\

\end{document}